\RequirePackage{amssymb}
\documentclass{article}
\usepackage{fullpage}
\usepackage{amsmath,amsthm}
\usepackage{graphicx}
\usepackage{multirow}
\usepackage{amsfonts}
\usepackage{xcolor}
\usepackage{enumitem}

\usepackage{xspace}
\usepackage{booktabs}

\usepackage{algorithm}
\usepackage{algcompatible}
\newcommand{\To}{\textbf{to}\xspace}

\newcommand{\algorithmicreturn}{\textbf{return}}
\newcommand{\RETURN}{\STATE\algorithmicreturn{}\xspace}
\algnewcommand{\IfThen}[2]{%
  \State \algorithmicif\ #1\ \algorithmicthen\ #2}
\algnewcommand{\IfThenEnd}[2]{%
  \State \algorithmicif\ #1\ \algorithmicthen\ #2\ \algorithmicend\ \algorithmicif}
\algnewcommand{\IfThenElse}[3]{%
  \State \algorithmicif\ #1\ \algorithmicthen\ #2\ \algorithmicelse\ #3}
\algnewcommand{\ForDoEnd}[3][]{%
  \ifthenelse{\equal{#1}{}}%
  {\State \algorithmicfor\ #2\ \algorithmicdo\ #3\ \algorithmicend\ \algorithmicfor}
  {\State\label{#1}\algorithmicfor\ #2\ \algorithmicdo\ #3\ \algorithmicend\ \algorithmicfor}
}

\algnewcommand\algorithmicreadonly{\textbf{Read-only:}}
\algnewcommand\READONLY{\item[\algorithmicreadonly]}%

\newcommand{\D}{\ensuremath{\mathbb{D}}\xspace}
\newcommand{\F}{\ensuremath{\mathbb{F}}\xspace}
\newcommand{\Z}{\ensuremath{\mathbb{Z}}\xspace}

\newcommand{\bnd}[2]{\ensuremath{#1\mathopen{}(#2)\mathclose{}}}
\newcommand{\bnddisplay}[2]{\ensuremath{#1\mathopen{}\left(#2\right)\mathclose{}}}

\newcommand{\bigO}[1]{\bnd{\mathcal{O}}{#1}}
\newcommand{\bigOdisplay}[1]{\bnddisplay{\mathcal{O}}{#1}}

\newcommand{\threshold}{\ensuremath{\text{Threshold}}}
\DeclareMathOperator{\me}{\,\ensuremath{\mathrel{\textrm{--}}=}\,}
\DeclareMathOperator{\pe}{\,\ensuremath{\mathrel{+}=}\,}
\DeclareMathOperator{\fe}{\,\ensuremath{\mathrel{{\ast}}=}\,}
\DeclareMathOperator{\de}{\,\ensuremath{\mathrel{/}=}\,}

\newcommand\DFT{\mathsf{DFT}}
\newcommand\brDFT{\mathsf{brDFT}}
\newcommand\brTFT{\mathsf{brTFT}}
\newcommand\partTFT{\mathsf{partTFT}}

\makeatletter
\DeclareRobustCommand{\cev}[1]{%
  {\mathpalette\do@cev{#1}}%
}
\newcommand{\do@cev}[2]{%
  \vbox{\offinterlineskip
    \sbox\z@{$\m@th#1 x$}%
    \ialign{##\cr
      \hidewidth\reflectbox{$\m@th#1\vec{}\mkern4mu$}\hidewidth\cr
      \noalign{\kern-\ht\z@}
      $\m@th#1#2$\cr
    }%
  }%
}
\makeatother
\newcommand{\mat}[1]{\textbf{\ensuremath{#1}}}
\newcommand{\Transpose}[1]{{{\mat{#1}}^{\intercal}}\xspace}
\newcommand{\gTranspose}[2]{{{\mat{\begingroup\setlength\arraycolsep{#1}#2\endgroup}}^{\intercal}}\xspace}

\usepackage{fancyvrb}

\newenvironment{smatrix}{\begin{bmatrix}}{\end{bmatrix}}
\def\matrixsize#1#2{{{#1}\times{#2}}}
\def\MatrixProduct#1#2{{\mat{#1}\cdot\mat{#2}}}
\def\triadone{brown}
\def\triadtwo{red}
\def\triadthree{blue}
\def\triadfour{green}
\def\triadfive{magenta}
\def\triadsix{gray}
\def\triadseven{violet}

\usepackage{bbding}
\newcommand{\xyes}{{\color{teal}\CheckmarkBold}}
\newcommand{\xno}{{\color{purple}\XSolidBrush}}

\newcommand{\LLTr}{\ensuremath{{\text{Low}}}}
\newcommand{\Low}[1]{\ensuremath{{\LLTr\left({\mat{#1}}\right)}}}

\newcommand{\MUL}{\normalfont{\textbf{MUL}}\xspace}
\newcommand{\ADD}{\normalfont{\textbf{ADD}}\xspace}
\newcommand{\SCA}{\normalfont{\textbf{SCA}}\xspace}
\title{In-place accumulation of fast multiplication formulae}

\author{Jean-Guillaume Dumas\footnote{
  {Universit\'e Grenoble Alpes}.
  {Laboratoire Jean Kuntzmann, CNRS, UMR 5224}.
  {150 place du Torrent, IMAG - CS 40700},
  {38058 Grenoble, cedex 9}
  {France}.
\href{mailto:Jean-Guillaume.Dumas@univ-grenoble-alpes.fr,Bruno.Grenet@univ-grenoble-alpes.fr}{\{firstname.lastname\}@univ-grenoble-alpes.fr}}
\and{Bruno Grenet}\footnotemark[1]}

\usepackage{color,svgcolor}
\usepackage[plainpages=true]{hyperref}

\makeatletter
\hypersetup{%
  pdftitle={\@title},
  pdfauthor={Jean-Guillaume Dumas and Bruno Grenet},
  breaklinks=true,
  colorlinks=true,
  linkcolor=purple,
  citecolor=blue,
  urlcolor=teal,
  hypertexnames=false,
}
\makeatother
\usepackage[capitalize]{cleveref}

\newcommand{\Algname}{Algorithm}
\newcommand{\Algsname}{{\Algname}s}
\crefname{algorithm}{\Algname}{\Algsname}
\Crefname{algorithm}{\Algname}{\Algsname}
\crefname{proposition}{Prop.}{Props.}
\Crefname{proposition}{Proposition}{Propositions}

\newtheorem{theorem}{Theorem}
\newtheorem{corollary}[theorem]{Corollary}
\newtheorem{proposition}[theorem]{Proposition}

\newtheorem{lemma}[theorem]{Lemma}

\newtheorem{remark}[theorem]{Remark}
\newtheorem{example}[theorem]{Example}
\begin{document}
\maketitle
\begin{abstract}
This paper deals with simultaneously fast and in-place algorithms
for formulae where the result has to be linearly accumulated:
some output variables are also input variables,
linked by a linear dependency.
Fundamental examples include the in-place accumulated multiplication
of polynomials or matrices, $C\pe{AB}$.
The difficulty is to combine in-place computations with fast
algorithms:
those usually come at the expense of (potentially large) extra
temporary space, but with accumulation the output variables are not
even available to store intermediate values.
We first propose a novel automatic design of fast and in-place
accumulating algorithms for any bilinear formulae (and thus for
polynomial and matrix multiplication) and then extend it to any linear
accumulation of a collection of functions.
For this, we relax the in-place model to any
algorithm allowed to modify its inputs, provided that those are
restored to their initial state afterwards.
This allows us, in fine, to derive unprecedented in-place accumulating
algorithms for fast polynomial multiplications and for Strassen-like
matrix multiplications.
\end{abstract}
\section{Introduction}
Multiplication is one of the most fundamental arithmetic operations in
computer science and in particular in computer algebra and symbolic
computation.
In terms of arithmetic operations, for instance, from the work
of~\cite{Karatsuba:1963:multiplication,SchoenhageStrassen1971,Strassen:1969:GENO}, many
sub-quadratic (resp. sub-cubic) algorithms
were developed for polynomial (resp. matrix) multiplication.
But these fast algorithms usually come at the expense
of (potentially large) extra temporary space to perform the
computation that could hinder their practical efficiency,
due for instance to cache misses.
On the contrary, classical, quadratic (resp. cubic)
algorithms, when computed sequentially, quite often require very few
(constant) extra registers.
Further work then proposed simultaneously ``fast'' and ``in-place''
algorithms, for matrix or polynomial
operations~\cite{jgd:2009:WinoSchedule,Roche:2009:spacetime,Harvey:2010:issactft,Giorgi:2019:issac:reductions,Giorgi:2020:issac:inplace}.

We here extend the latter line of work for
\emph{accumulating} algorithms.
Actually, one of the main ingredients on non-accumulating algorithms is to
use the (free) space of the output as intermediate storage.
But when the result has to be accumulated, \emph{i.e.}, if the output is also
part of the input, this free space does not even exist.
To be able to design accumulating in-place algorithms we thus relax the in-place
model to allow algorithms to also modify their input, therefore to use
them as intermediate storage,
\emph{provided that they are restored to their initial state after
  completion of the procedure}.
This is in fact a natural possibility in many programming environments.
Furthermore, this restoration allows for recursive combinations of
such procedures, as the (non-concurrent) recursive calls will not
mess up the state of their callers.
We thus propose a generic technique transforming any bilinear
algorithm into an in-place algorithm under this model.
This directly applies to accumulating polynomial and matrix
multiplication algorithms, including fast ones.
Further, the technique actually generalizes to any linear
accumulation, \emph{i.e.}, not only bilinear formulae, provided that the input
of the accumulation can be itself reversibly computed in-place
(therefore also potentially in-place of some of its own input if
needed).

Next, we give our model for in-place computations
and recall classical in-place algorithms
in~\cref{ssec:inplace}.
We then detail in~\cref{sec:linacc} our novel technique for in-place
accumulation.
Finally, we apply this technique and further optimizations in order to
derive new fast and in-place algorithms for the accumulating
multiplication of matrices, \cref{sec:strassen}, and of polynomials,
\cref{sec:inpaccpol}.

\section{Computational model}\label{ssec:inplace}
Our computational model is an \emph{algebraic RAM}. Inputs
and outputs are arrays of ring elements.
(For simplicity, our algorithms are described over a
finite field $\F$, unless otherwise stated.)
The machine is made of \emph{algebraic registers} that each contain
one ring element, and \emph{pointer registers} that each contain
one pointer, that is one integer. Atomic operations are ring
operations on the algebraic registers and basic pointer arithmetic.
We assume that the pointer registers are large enough to store the
length of the input/output arrays.

Both inputs and outputs have read/write permissions.
But algorithms are only allowed to modify their inputs
\textbf{if their inputs are restored to their initial
state} afterwards.
In this model, we call \emph{in-place} an algorithm using only
\textbf{the space of its inputs, its outputs, and at most $\bigO{1}$ extra
space}. For recursive algorithms, some space may be required to store the recursive
call stack. (This stack is only made of pointers and its size is bounded
by the recursion depth of the algorithms. In practice, it is managed by the compiler.)
Nonetheless, we call \emph{in-place} a recursive algorithm whose only
extra space is the call stack.
In our complexity summaries (\Cref{tab:kara,tab:fft}), we include the
size of the stack.

The main limitations of this model are for black-box inputs, or for
inputs whose representations share some data.
A model with read-only inputs would be more powerful, but mutable
inputs turn out to be necessary in our case.
In particular, the algorithms we describe are \emph{in-place
with accumulation}. The archetypical example is a multiply-accumulate
operation $a \pe b\times c$. For such an algorithm, the condition is
that $b$ and $c$ are restored to their initial states at the end of
the computation, while $a$ (which is also part of the input) is replaced
by $a+bc$.
As a variant, we describe \emph{over-place} algorithms, that %
replace (parts of) the input by the output (e.g., $\vec{a}\gets b{\cdot}\vec{a}$).
Similarly, all of the input can be modified, provided that the
parts of the input that are not the output are restored afterwards.
In the following we signal by a ``\algorithmicreadonly'' tag the parts
of the input that the algorithm is not allowed to modify (the other
parts are modifiable as long as they are restored).
Note that in-place algorithms with accumulation are a special case of
over-place algorithms.
Our model is somewhat similar to catalytic machines and transparent
space~\cite{Buhrman:2014:STOC:catalytic}, but %
using only the input and output as catalytic space. Also, we do preserve the
(not in-place) time complexity, up to a (quasi)-linear overhead.
We refer to~\cite{Buhrman:2014:STOC:catalytic,Roche:2009:spacetime,Giorgi:2019:issac:reductions}
for more details. %

Classical algorithms for matrix or polynomial operations can be
performed in-place, without any call stack, as recalled
in~\cref{alg:classicmul}.
\begin{algorithm}[htbp]\caption{Quadratic/cubic in-place accumulating multiplications.}\label{alg:classicmul}
\begin{minipage}{.475\columnwidth}
\begin{algorithmic}[1]
\REQUIRE $A$, $B$, $C$, deg.  $m$, $n$, $m+n$.
\READONLY Polynomials $A$, $B$.
\ENSURE $C(X)\pe{A(X)B(X)}$
\FOR{$i$, $j$}
\STATE $C[i{+}j]\pe{A[i]B[j]}$;
\ENDFOR
\end{algorithmic}
\end{minipage}\hfill
\begin{minipage}{.475\columnwidth}
\begin{algorithmic}[1]
\REQUIRE $A$, $B$, $C$, $m{\times}\ell$,
$\ell{\times}n$, $m{\times}n$.
\READONLY Matrices $A$, $B$.
\ENSURE $C\pe{AB}$
\FOR{$i$, $j$, $k$}
\STATE $C_{ij}\pe{A_{ik}B_{kj}}$;
\ENDFOR
\end{algorithmic}
\end{minipage}
\end{algorithm}

\section{In-place linear accumulation}\label{sec:linacc}
Karatsuba polynomial
multiplication~\cite{Karatsuba:1963:multiplication}
and Strassen matrix multiplication~\cite{Strassen:1969:GENO}
are famous optimizations of bilinear formulae on their inputs: results
are linear combinations of products of bilinear combinations of the inputs.
To compute recursively such a formula in-place, we perform
each product one at a time. For each product, both factors are then
linearly combined in-place into one of the entry beforehand and
restored afterwards. The product of both entries is at that point
accumulated in one part of the output and then distributed to the
other parts.
The difficulty is to perform this distribution in-place, {\em without
recomputing the product}. Our idea is to pre-subtract one output from the
other, then accumulate the product to one output, and finally re-add the
newly accumulated output to the other one: overall both outputs just
have accumulated the product, in-place. Potential constant factors can
also be dealt with pre-divisions and post-multiplications.
Basically we need two kinds of in-place operations, and their
combinations.
First, as shown in~\cref{eq:basemul}, an in-place accumulation of a
quantity multiplied by a (known in advance) invertible constant:
\begin{equation}\label{eq:basemul}
\left\lbrace{}c\de\mu;~c\pe m;~c\fe\mu;\right\rbrace~\text{computes
  in-place}~c\gets{c+\mu\cdot{m}}.
\end{equation}
Second, as shown in~\cref{eq:basedist}, an in-place distribution of
the same quantity, without recomputation, to several outputs:
\begin{equation}\label{eq:basedist}
\left\lbrace{}d\me{c};~c\pe{m};~d\pe{c};\right\rbrace~\text{computes
  in-place}~\left\{\begin{array}{@{}l@{\,}l@{}}
c&\gets{c+m};\\
d&\gets{d+m}.\\
\end{array}
\right.
\end{equation}

\Cref{ex:bilin} shows how to combine several of these operations,
while also linearly combining parts of the input.
\begin{example}\label{ex:bilin}
  Suppose that for some inputs/outputs $a$, $b$, $c$, $d$, $r$, $s$, one wants to compute an
  intermediate product $p=(a+3b)*(c+d)$ only once and then distribute
  and accumulate that product to two of its outputs (or results),
  such that we have both $r\gets{r+5p}$ and $s\gets{s+2p}$.
  To perform this in-place, first accumulate $a\pe{3b}$ and $c\pe{d}$,
  then pre-divide $r$ by $5$, as in~\cref{eq:basemul}.
  Now we directly have $p=ac$ and it can be computed once,
  and then accumulated to $r$, and to $s$, if the latter is prepared:
  divide it by $2$, and pre-subtract $r$ or, equivalently,
  pre-subtract $2r$. This is
  $s\me{2r}$ followed by $r\pe{ac}$. After this, we can
  reciprocate (or unroll) the precomputations: this distributes
  the product to the other result and restores the read-only inputs to
  their initial state.
  This is summarized as:\\
  \begin{center}\fbox{%
      \ensuremath{\ \left\lbrace\begin{aligned}
	    a\pe{3b}; &\quad c\pe{d}; & r\de{5}\phantom{r};\\[-5pt]
	    s\me{2r}; &\quad r\pe{ac}; & s\pe{2r}; \\[-5pt]
	     a\me{3b};&\quad c\me{d}; &r\fe{5}\phantom{r};
	  \end{aligned}\right\rbrace
	\begin{array}{l}
	  \text{computes in-place:}\\[-0pt]
	  \left\lbrace\begin{aligned}
	    r& \gets{r+5(a+3b)(c+d)};\\[-5pt]
	    s& \gets{s+2(a+3b)(c+d)}.\\[-1pt]
	  \end{aligned}\right.
	\end{array}
      }}\end{center}
\end{example}

\Cref{alg:bilin} shows how to implement this in general, taking into
account the constant (or read-only) multiplicative coefficients of all
the linear combinations. We suppose that inputs are in three distinct
sets: left-hand sides, $\vec{a}$, right-hand sides, $\vec{b}$, and those
accumulated to the results, $\vec{c}$.
We denote by $\odot$ the point-wise multiplications of
left-hand sides by right-hand sides.
Then~\cref{alg:bilin} computes $\vec{c}\pe\mat{\mu}\vec{m}$, for
$\vec{m}=(\mat{\alpha}\vec{a})\odot(\mat{\beta}\vec{b})$, with
$\mat{\alpha}$, $\mat{\beta}$ and $\mat{\mu}$ matrices of
constants.

\begin{algorithm}[htbp]
  \caption{In-place bilinear formula.}\label{alg:bilin}
  \begin{algorithmic}[1]
    \REQUIRE $\vec{a}\in\F^m$, $\vec{b}\in\F^n$, $\vec{c}\in\F^s$;
    $\mat{\alpha}\in\F^{t{\times}m}$, $\mat{\beta}\in\F^{t{\times}n}$,
    $\mat{\mu}\in\F^{s{\times}t}$.
    \READONLY$\mat{\alpha}$, $\mat{\beta}$, $\mat{\mu}$ (all $3$ without zero-rows).
    \ENSURE $\vec{c}\pe\mat{\mu}\vec{m}$, for
    $\vec{m}=(\mat{\alpha}\vec{a})\odot(\mat{\beta}\vec{b})$.
    \FOR{$\ell=1$ \To $t$}
    \STATE\label{lin:alpha}Find one~$i$ s.t. $\alpha_{\ell,i}\neq{0}$;
    $a_i\fe\alpha_{\ell,i}$;
    \ForDoEnd[lin:foralpha]{$\lambda=1$ \To $m$, $\lambda\neq{i}$,
      $\alpha_{\ell,\lambda}\neq{0}$ }{ $a_i\pe\alpha_{\ell,\lambda}a_\lambda$}
    \STATE\label{lin:beta}Find one~$j$ s.t. $\beta_{\ell,j}\neq{0}$;
    $b_j\fe\beta_{\ell,j}$;
    \ForDoEnd[lin:forbeta]{$\lambda=1$ \To $n$, $\lambda\neq{j}$,
      $\beta_{\ell,\lambda}\neq{0}$}{ $b_j\pe\beta_{\ell,\lambda}b_\lambda$}
    \STATE\label{lin:mu}Find one~$k$ s.t. $\mu_{k,\ell}\neq{0}$;
    $c_k\de\mu_{k,\ell}$;
    \ForDoEnd[lin:formu]{$\lambda=1$ \To $s$, $\lambda\neq{k}$,
      $\mu_{\lambda,\ell}\neq{0}$}{$c_\lambda\me\mu_{\lambda,\ell}{c_k}$}
    \STATE\label{lin:product}$c_k\pe{a_i\cdot{b_j}}$\hfill\COMMENT{This is the product $m_\ell$, computed only once}
    \ForDoEnd[lin:distribmu]{$\lambda=1$ \To $s$, $\lambda\neq{k}$,
      $\mu_{\lambda,\ell}\neq{0}$}{$c_\lambda\pe\mu_{\lambda,\ell}{c_k}$}
    \hfill\COMMENT{undo~\ref{lin:formu}}
    \STATE\label{lin:fmuk}$c_k\fe\mu_{k,\ell}$;\hfill\COMMENT{undo~\ref{lin:mu}}
    \ForDoEnd[lin:distribbeta]{$\lambda=1$ \To $n$, $\lambda\neq{j}$,
      $\beta_{\ell,\lambda}\neq{0}$}{$b_j\me\beta_{\ell,\lambda}b_\lambda$}
    \hfill\COMMENT{undo~\ref{lin:forbeta}}
    \STATE\label{lin:dbetaj}$b_j\de\beta_{\ell,j}$;\hfill\COMMENT{undo~\ref{lin:beta}}
    \ForDoEnd[lin:distribalpha]{$\lambda=1$ \To $m$, $\lambda\neq{i}$,
      $\alpha_{\ell,\lambda}\neq{0}$}{$a_i\me\alpha_{\ell,\lambda}a_\lambda$}
    \hfill\COMMENT{undo~\ref{lin:foralpha}}
    \STATE\label{lin:dalphai}$a_i\de\alpha_{\ell,i}$;\hfill\COMMENT{undo~\ref{lin:alpha}}
    \ENDFOR
    \RETURN $\vec{c}$.
  \end{algorithmic}
\end{algorithm}

\begin{remark}\label{rk:parallelism}
  \Cref{lin:alpha,lin:foralpha,lin:beta,lin:forbeta,lin:mu,lin:formu,lin:distribmu,lin:fmuk,lin:distribbeta,lin:dbetaj,lin:distribalpha,lin:dalphai}   of~\cref{alg:bilin} are acting on independent parts of the
  input, $\vec{a}$ and $\vec{b}$, and of the output $\vec{c}$.
  If needed they could therefore be computed
  in parallel
  or in different orders,
  and even potentially grouped or factorized across the main loop (on $\ell$).
\end{remark}

To simplify the counting of operations,
we denote by \ADD both the addition or subtraction of elements, $\pe$
or $\me$; by \MUL the (tensor) product of elements, $\odot$;
and by \SCA the scaling by constants, $\fe$ or $\de$.
We also denote by $\#x$ (resp. $\sharp{x}$) the number of
non-zero (resp. $\not\in\{0,1,-1\}$) elements in a matrix $x$.

\begin{theorem}\label{thm:bilin}
\Cref{alg:bilin} is correct, in-place, and requires
$t$ \MUL,
$2(\#\alpha+\#\beta+\#\mu)-5t$ \ADD and
$2(\sharp\alpha+\sharp\beta+\sharp\mu)$ \SCA operations.
\end{theorem}
\begin{proof}
First, as the only used operations ($\pe$, $\me$, $\fe$, $\de$) are
in-place ones, the algorithm is in-place.
Second, the algorithm is correct both for the input and the
output:
the input is well restored, as
$(\alpha_{\ell,i}a_i+\sum\alpha_{\ell,\lambda}a_\lambda-\sum\alpha_{\ell,\lambda}a_{\lambda})/\alpha_{\ell,i}=a_i$
and
$(\beta_{\ell,j}b_j+\sum\beta_{\ell,\lambda}b_\lambda-\sum\beta_{\ell,\lambda}b_\lambda)/\beta_{\ell,j}=b_j$;
the output is correct as
$c_\lambda-\mu_{\lambda,\ell}c_k/\mu_{k,\ell}+\mu_{\lambda,\ell}(c_k/\mu_{k,\ell}+a_ib_j)=c_\lambda+\mu_{\lambda,\ell}a_ib_j$
and
$(c_k/\mu_{k,\ell}+a_ib_j)\mu_{k,\ell}=c_k+\mu_{k,\ell}a_ib_j$.
Third, for the number of operations,
\cref{lin:alpha,lin:foralpha} require one multiplication by a constant for each
non-zero element $a_{\lambda}$ in the row and one less addition.
But multiplications and divisions by $1$ are no-op, and by $-1$ can be
dealt with subtraction. This is $\#\alpha-t$ additions
and $\sharp\alpha$ constant multiplications.
\cref{lin:beta,lin:forbeta} (resp. \cref{lin:mu,lin:formu}) are
similar for each non-zero element in $b_{\lambda}$ (resp. in
$\mu$). Finally, \cref{lin:product} performs $t$ multiplications of
elements and $t$ additions. The remaining lines double the number of
\ADD and \SCA.
This is $t+2(\#\alpha+\#\beta+\#\mu-3t)=2(\#\alpha+\#\beta+\#\mu)-5t$ \ADD.
\end{proof}

\begin{remark} Similarly, slightly more generic accumulation
  operations of the form
  $\vec{c}\gets\vec{\gamma}\odot\vec{c}+\mat{\mu}\vec{m}$, for a
  vector $\gamma\in\F^{s}$, can also be computed in-place: precompute
  first $\vec{c}\gets\vec{\gamma}\odot\vec{c}$, then
  call~\cref{alg:bilin}.
\end{remark}

For instance, to use~\cref{alg:bilin} with matrices or polynomials,
each product $m_\ell$ is in fact computed recursively.
Further, in an actual implementation of a fixed formula, one can
combine more efficiently the pre- and post-computations over
the main loop on $\ell$, as in~\cref{rk:parallelism}.
See~\cref{sec:strassen,sec:inpaccpol} for examples of recursive
calls, together with sequential optimizations and combinations.

In fact the method for accumulation, computing each bilinear
multiplication once is generalizable.
With the notations of~\cref{alg:bilin}, any algorithm of the form
$\vec{c}\pe\mat{\mu}\vec{m}$ can benefit from this technique,
provided that each $m_j$ can be obtained from a function
that can be computed in-place.
Let $F_j:\Omega\to\F$ be such a function on some inputs from a
space $\Omega$, for which an in-place algorithm exists.
Then we can accumulate it in-place, \emph{if it satisfies the following
constraint}, that it is not using its output space as an available
intermediary memory location.
Further, this function can be in-place in different models:
it can follow our model of~\cref{ssec:inplace}, if there is a way to
put its input back into their initial states, or some other model,
again provided that it follows the above constraint.
Then, the idea is just to keep from~\cref{alg:bilin}
the~\cref{lin:mu,lin:formu,lin:product,lin:distribmu,lin:fmuk},
replacing~\cref{lin:product} by the in-place call to $F_j$,
potentially surrounding that call by manipulations on the inputs of
$F_j$ (just like the one performed on $\vec{a}$ and $\vec{b}$
in~\cref{alg:bilin}).
We give examples of the application of the generalized method
of~\cref{thm:general} to non-bilinear formulae in~\cref{app:aat},
and we can thus show that:
\begin{theorem}\label{thm:general}
Let $\vec{c}\in\F^s$ and $\mat{\mu}\in\F^{s{\times}t}$, without zero-rows.
Let $\vec{F}=(F_j:\Omega\to\F)_{j=1..t}$ be a collection of functions
and $\omega\in\Omega$.
If all these functions are computable in-place, without using their
output space as an %
intermediary memory location,
then there exists an in-place algorithm computing
$\vec{c}\pe\mat{\mu}\vec{F}(\omega)$ in-place, requiring a single call to
each $F_j$,
together with
$(2\#\mu-t)$ \ADD and
$2\sharp\mu$ \SCA ops.
\end{theorem}

\section{In-place Strassen matrix multiplication with accumulation}\label{sec:strassen}
\subsection{7 recursive calls and 18 additions}\label{app:inplsw}

Considered as~${\matrixsize{2}{2}}$ matrices, the matrix
product with accumulation ~${\mat{C}\pe\MatrixProduct{A}{B}}$ could be computed using
Strassen-Winograd (S.-W.) algorithm by performing the following computations: %
\begin{gather}
\begin{array}{ll}
\mathcolor{\triadone}{\rho_{1}}\gets{\mathcolor{\triadone}{a_{11}}\mathcolor{\triadone}{b_{11}}},
\quad
\mathcolor{\triadthree}{\rho_{3}}\gets{(\mathcolor{\triadthree}{-a_{11}-a_{12}+a_{21}+a_{22}})\mathcolor{\triadthree}{b_{22}}},
\\
\mathcolor{\triadtwo}{\rho_{2}}\gets{\mathcolor{\triadtwo}{a_{12}}\mathcolor{\triadtwo}{b_{21}}},
\quad
\mathcolor{\triadfour}{\rho_{4}}\gets{\mathcolor{\triadfour}{a_{22}}(\mathcolor{\triadfour}{-b_{11}+b_{12}+b_{21}-b_{22}})},
\\
\mathcolor{\triadfive}{\rho_{5}}\gets{(\mathcolor{\triadfive}{a_{21}+a_{22}})(\mathcolor{\triadfive}{-b_{11}+b_{12}})},
\quad
\mathcolor{\triadsix}{\rho_{6}}\gets{(\mathcolor{\triadsix}{-a_{11}+a_{21}})(\mathcolor{\triadsix}{b_{12}-b_{22}})},
\\
\mathcolor{\triadseven}{\rho_{7}}\gets{(\mathcolor{\triadseven}{-a_{11}+a_{21}+a_{22}})(\mathcolor{\triadseven}{-b_{11}+b_{12}-b_{22}})},
\end{array}\nonumber\\
\label{eq:StrassenWinogradMultiplicationAlgorithm}
\begin{smatrix} c_{11} &c_{12} \\ c_{21} &c_{22} \end{smatrix}
\pe
\begin{smatrix}
\mathcolor{\triadone}{\rho_{1}} + \mathcolor{\triadtwo}{\rho_{2}} &
\mathcolor{\triadone}{\rho_{1}} - \mathcolor{\triadthree}{\rho_{3}} + \mathcolor{\triadfive}{\rho_{5}} - \mathcolor{\triadseven}{\rho_{7}}\\
\mathcolor{\triadone}{\rho_{1}} + \mathcolor{\triadfour}{\rho_{4}} +
\mathcolor{\triadsix}{\rho_{6}} - \mathcolor{\triadseven}{\rho_{7}} &
\mathcolor{\triadone}{\rho_{1}}+\mathcolor{\triadfive}{\rho_{5}} +
\mathcolor{\triadsix}{\rho_{6}} - \mathcolor{\triadseven}{\rho_{7}}
\end{smatrix}.
\end{gather}
This algorithm uses $7$ multiplications of half-size matrices and $24+4$
additions (that can be factored into only $15+4$~\cite{Winograd:1977:complexite}:
$4$ involving $A$, $4$ involving $B$ and $7$ involving the products,
plus $4$ for the accumulation).
This can be used recursively on matrix blocks, halved at each
iteration, to obtain a sub-cubic algorithm. To save on operations, it
is of course interesting to compute the products only once, that is
store them in extra memory chunks.
To date, up to our knowledge, the best versions that reduced this
extra memory space (also overwriting the
input matrices but not putting them back in place) were proposed
in~\cite{jgd:2009:WinoSchedule}:
their best sub-cubic accumulating product used $2$ temporary blocks
per recursive level, thus a total of extra memory required to
be~$\frac{2}{3}n^2$.
With~\cref{alg:bilin} we instead obtain an in-place sub-cubic algorithm for
accumulating matrix multiplication, without extra temporary field element.
From~\cref{eq:StrassenWinogradMultiplicationAlgorithm} indeed (see
also the representation
in~\cite{hopcroft:1973,Bshouty:1995:minwinoadd}, denoted \textsc{HM}),
we can extract
the matrices %
\begin{equation}\label{eq:alphabetamu}
\mu=\begin{smatrix}
\mathcolor{\triadone}{1}&\mathcolor{\triadtwo}{1}&\mathcolor{\triadthree}{0}&\mathcolor{\triadfour}{0}&\mathcolor{\triadfive}{0}&\mathcolor{\triadsix}{0}&\mathcolor{\triadseven}{0}\\
\mathcolor{\triadone}{1}&\mathcolor{\triadtwo}{0}&\mathcolor{\triadthree}{-1}&\mathcolor{\triadfour}{0}&\mathcolor{\triadfive}{1}&\mathcolor{\triadsix}{0}&\mathcolor{\triadseven}{-1}\\
\mathcolor{\triadone}{1}&\mathcolor{\triadtwo}{0}&\mathcolor{\triadthree}{0}&\mathcolor{\triadfour}{1}&\mathcolor{\triadfive}{0}&\mathcolor{\triadsix}{1}&\mathcolor{\triadseven}{-1}\\
\mathcolor{\triadone}{1}&\mathcolor{\triadtwo}{0}&\mathcolor{\triadthree}{0}&\mathcolor{\triadfour}{0}&\mathcolor{\triadfive}{1}&\mathcolor{\triadsix}{1}&\mathcolor{\triadseven}{-1}
\end{smatrix},\quad
\alpha=\begin{smatrix}
\mathcolor{\triadone}{1}&\mathcolor{\triadone}{0}&\mathcolor{\triadone}{0}&\mathcolor{\triadone}{0}\\
\mathcolor{\triadtwo}{0}&\mathcolor{\triadtwo}{1}&\mathcolor{\triadtwo}{0}&\mathcolor{\triadtwo}{0}\\
\mathcolor{\triadthree}{-1}&\mathcolor{\triadthree}{-1}&\mathcolor{\triadthree}{1}&\mathcolor{\triadthree}{1}\\
\mathcolor{\triadfour}{0}&\mathcolor{\triadfour}{0}&\mathcolor{\triadfour}{0}&\mathcolor{\triadfour}{1}\\
\mathcolor{\triadfive}{0}&\mathcolor{\triadfive}{0}&\mathcolor{\triadfive}{1}&\mathcolor{\triadfive}{1}\\
\mathcolor{\triadsix}{-1}&\mathcolor{\triadsix}{0}&\mathcolor{\triadsix}{1}&\mathcolor{\triadsix}{0}\\
\mathcolor{\triadseven}{-1}&\mathcolor{\triadseven}{0}&\mathcolor{\triadseven}{1}&\mathcolor{\triadseven}{1}
\end{smatrix},\quad
\beta=\begin{smatrix}
\mathcolor{\triadone}{1}&\mathcolor{\triadone}{0}&\mathcolor{\triadone}{0}&\mathcolor{\triadone}{0}\\
\mathcolor{\triadtwo}{0}&\mathcolor{\triadtwo}{0}&\mathcolor{\triadtwo}{1}&\mathcolor{\triadtwo}{0}\\
\mathcolor{\triadthree}{0}&\mathcolor{\triadthree}{0}&\mathcolor{\triadthree}{0}&\mathcolor{\triadthree}{1}\\
\mathcolor{\triadfour}{-1}&\mathcolor{\triadfour}{1}&\mathcolor{\triadfour}{1}&\mathcolor{\triadfour}{-1}\\
\mathcolor{\triadfive}{-1}&\mathcolor{\triadfive}{1}&\mathcolor{\triadfive}{0}&\mathcolor{\triadfive}{0}\\
\mathcolor{\triadsix}{0}&\mathcolor{\triadsix}{1}&\mathcolor{\triadsix}{0}&\mathcolor{\triadsix}{-1}\\
\mathcolor{\triadseven}{-1}&\mathcolor{\triadseven}{1}&\mathcolor{\triadseven}{0}&\mathcolor{\triadseven}{-1}
\end{smatrix}.
\end{equation}

All coefficients being $1$ or $-1$ the resulting in-place algorithm
can of course compute the accumulation $C\pe{AB}$ without constant
multiplications.
It thus requires $7$ recursive calls and, from~\cref{thm:bilin},
at most $2(\#\alpha+\#\beta+\#\mu-3t)=2(14+14+14-3*7)=42$ block
additions.
Just like the $24$ additions
of~\cref{eq:StrassenWinogradMultiplicationAlgorithm} can be factored
into $15$, one can optimize also the in-place algorithm.
For instance, looking at $\alpha$ we see that performing the products
in the order $\rho_{6}$, $\rho_{7}$, $\rho_{3}$, $\rho_{5}$ and
accumulating in $a_{21}$ enables to perform all additions/subtractions
in $A$ with only $6$ operations (this is in fact optimal, see~\cref{prop:six}).
This is similar for $\beta$ if the order
$\rho_{6}$, $\rho_{7}$, $\rho_{4}$, $\rho_{5}$ is used and
accumulation is in $b_{12}$.
Thus ordering for instance $\rho_{6}$, $\rho_{7}$, $\rho_{4}$,
$\rho_{3}$, $\rho_{5}$ will reduce the number of block additions to $26$.
Now looking at $\mu$ (more precisely at its transpose,
see~\cite{Kaminski:1988:transpose}), a similar reduction can be
obtained, e.g., if one of the orders
($\rho_{6}$, $\rho_{7}$, $\rho_{1}$, $\rho_{5}$)
or
($\rho_{5}$, $\rho_{7}$, $\rho_{1}$, $\rho_{6}$) is used
and accumulation is in $c_{22}$.

Therefore, using the ordering
$\rho_{6},\rho_{7},\rho_{1},\rho_{4},\rho_{3},\rho_{5},\rho_{2}$
requires only $18$ additions (plus $7$ accumulations in $C$), as shown
with~\cref{alg:ipsw}.
Thus, without thresholds and for powers of two, the dominant term of
the overall arithmetic cost is $8n^{\log_2(7)}$, for the
in-place version,
roughly a third more operations than the $6n^{\log_2(7)}$ dominant
term of the  cost for the version using extra temporaries.

We here give an in-place version of Strassen-Winograd algorithm for
matrix multiplication.
We first directly apply our~\cref{alg:bilin} to the classical, not
in-place Strassen-Winograd algorithm, following the specific scheduling
strategy of~\cref{sec:strassen}. This strategy enables to reduce the
number of additions obtained when calling~\cref{alg:bilin}, from $42+7$
to $18+7$: mostly remove successive additions/subtractions that are
reciprocal on either sub-matrices. This optimized version is given
in~\cref{alg:ipsw} and reaches the minimal possible number of
extra additions/subtractions, as shown in~\cref{thm:eighteen}.

\begin{algorithm}[ht]
\caption{In-place accumulating S.-W. matrix-multiplication.}\label{alg:ipsw}
\begin{algorithmic}
\REQUIRE
$A=\begin{smatrix} a_{11} & a_{12} \\  a_{21} &a_{22}\end{smatrix}$,
$B=\begin{smatrix} b_{11} & b_{12} \\  b_{21} &b_{22}\end{smatrix}$,
$C=\begin{smatrix} c_{11} & c_{12} \\  c_{21} &c_{22}\end{smatrix}$.
\ENSURE $C\pe{AB}$.
\end{algorithmic}
\begin{Verbatim}[commandchars=\\\{\},codes={\catcode`$=3\catcode`^=7\catcode`_=8}]
A$_{21}$ := A$_{21}$ - A$_{11}$; B$_{12}$ := B$_{12}$ - B$_{22}$; C$_{21}$ := C$_{21}$ - C$_{22}$;
\colorbox{cyan!10}{C$_{22}$ := C$_{22}$ + A$_{21}$ * B$_{12}$;}
A$_{21}$ := A$_{21}$ + A$_{22}$; B$_{12}$ := B$_{12}$ - B$_{11}$; C$_{12}$ := C$_{12}$ - C$_{22}$;
\colorbox{cyan!10}{C$_{22}$ := C$_{22}$ - A$_{21}$ * B$_{12}$;}
C$_{11}$ := C$_{11}$ - C$_{22}$;
\colorbox{cyan!10}{C$_{22}$ := C$_{22}$ + A$_{11}$ * B$_{11}$;}
C$_{11}$ := C$_{11}$ + C$_{22}$; B$_{12}$ := B$_{12}$ + B$_{21}$; C$_{21}$ := C$_{21}$ + C$_{22}$;
\colorbox{cyan!10}{C$_{21}$ := C$_{21}$ + A$_{22}$ * B$_{12}$;}
B$_{12}$ := B$_{12}$ + B$_{22}$; B$_{12}$ := B$_{12}$ - B$_{21}$; A$_{21}$ := A$_{21}$ - A$_{12}$;
\colorbox{cyan!10}{C$_{12}$ := C$_{12}$ - A$_{21}$ * B$_{22}$;}
A$_{21}$ := A$_{21}$ + A$_{12}$; A$_{21}$ := A$_{21}$ + A$_{11}$;
\colorbox{cyan!10}{C$_{22}$ := C$_{22}$ + A$_{21}$ * B$_{12}$;}
C$_{12}$ := C$_{12}$ + C$_{22}$; B$_{12}$ := B$_{12}$ + B$_{11}$; A$_{21}$ := A$_{21}$ - A$_{22}$;
\colorbox{cyan!10}{C$_{11}$ := C$_{11}$ + A$_{12}$ * B$_{21}$;}
\end{Verbatim}
\end{algorithm}
The number of temporary blocks of dimensions $\frac{n}{2}\times\frac{n}{2}$
required for the computation in~\cref{alg:ipsw} is compared to that of
previously known algorithms in~\cref{tab:inpsw}.
\begin{table}[ht]\centering
\caption{Reduced-memory accumulating S.-W. multiplication.}\label{tab:inpsw}
\begin{tabular}{lccc}
\toprule
Alg. & Temp. blocks & inputs & accumulation \\
\midrule
\cite{HussLederman:1996:ISA} & $3$ & {\color{teal} read-only} & {\xyes}\\
\cite{jgd:2009:WinoSchedule} & $2$ & {\color{teal} read-only} & {\xyes}\\
\cref{alg:ipsw} & $0$ & mutable & {\xyes}\\
\bottomrule
\end{tabular}
\end{table}

We now prove that $18$ additions is the minimal number of additions
required by an in-place algorithm resulting
from any bilinear algorithm for matrix multiplication using only $7$
multiplications.
For this we consider elementary operations on
variables (similar to elementary linear algebra operators):
\emph{variable-switching} (swapping variable $i$ and variable $j$);
\emph{variable-multiplying} (multiplying a variable by a constant);
\emph{variable-addition} (adding one variable, potentially multiplied
by a constant, to another variable).
An \emph{elementary program} is a program using only these
three kind of operations.
Now, the in-place implementation of a linear function on
its input, for $\mat{\alpha}\in\F^{t{\times}m}$ and $\vec{a}\in\F^m$,
is the computation of each of the $t$ coefficients of
$\mat{\alpha}\vec{a}$, using only elementary operations and
only the variables of $\vec{a}$ as temporary variables.
We start by proving in~\cref{lem:nocannorzero}
that in any bilinear algorithm for matrix
multiplication using only $7$ multiplications, the columns of the
associated matrices $\mat{\alpha},\mat{\beta},\mat{\mu}$
(as in~\cref{eq:alphabetamu})
cannot contain too many zeroes.
\begin{lemma}\label{lem:nocannorzero}
If
$(\mat{\alpha},\mat{\beta},\mat{\mu})\in\F^{7{\times}4}\times\F^{7{\times}4}\times\F^{4{\times}7}$
is the \textsc{HM} representation of a bilinear algorithm for matrix multiplication, then none of
$\mat{\alpha},\mat{\beta},\Transpose{\mat{\mu}}$
contains a zero column vector, nor a multiple of a standard basis vector.
\end{lemma}
\begin{proof}
The dimensions of the matrices indicate that the multiplicative
complexity of the algorithm is $7$.
From~\cite{Groote:1978:optimal} we know that all such bilinear
algorithms can be obtained from one another.
Following~\cite[Lemma~6]{Bshouty:1995:minwinoadd}, then any associated
$\mat{\alpha},\mat{\beta},\Transpose{\mat{\mu}}$ matrix is some row or column
permutation, or the multiplication by some $G\otimes{H}$ (the Kronecker
product of two invertible $2{\times}2$ matrices), of the matrices
of~\cref{eq:alphabetamu}.
By duality~\cite{Hopcroft:1973:duality}, see
also~\cite[Eq. (3)]{Bshouty:1995:minwinoadd},
it is also sufficient to consider any one of the $3$ matrices.
We thus let
$K=G\otimes{H}$. %
Then any column of $K$ is of the form
$\Transpose{\begin{bmatrix} ux , uy , vx , vy\end{bmatrix}}$,
where $\begin{smatrix} u\\v\end{smatrix}$ is a column of $G$ and
$\begin{smatrix}x\\y\end{smatrix}$ is a column of $H$.
Further as $G$ is invertible,
$u$ and $v$
cannot be both zero simultaneously
and, similarly,
$x$ and $y$
cannot be both zero simultaneously.
Now consider for instance
$\mat{\alpha}\cdot{K}$, with
$\mat{\alpha}$ of~\cref{eq:alphabetamu}.
Then any column $\vec{\theta}$ of $\mat{\alpha}\cdot{K}$
is of the form:\\
\(\gTranspose{1.5pt}{\begin{bmatrix}ux, uy, -ux-uy+vx+vy, vy, vx+vy, -ux+vx,-ux+vx+vy\end{bmatrix}}\)\!.\\
For such a column to be a multiple of a standard basis vector
or the zero vector, at least $6$ of its $7$ coefficients must be zero.
For instance, this means that at least two out of rows $1$, $2$ and $4$
must be zero: or that at least two of $ux$, $uy$ or $vy$ must be
zero. This limits us to three cases: (1) $u=0$, (2) $y=0$ or (3) $x=v=0$.
If $u=0$, then
$\vec{\theta}=v\gTranspose{1pt}{\begin{bmatrix}0,0,x+y,y,x+y,x,x+y\end{bmatrix}}$;
at least one of rows $4$ or $6$ has to be zero, thus,
w.l.o.g. suppose $x=0$, we obtain that
$\vec{\theta}=vy\gTranspose{1pt}{\begin{bmatrix}0,0,1,1,1,0,1\end{bmatrix}}$
with none of $v$ nor $y$ being zero (otherwise $G$ or $H$ is not
invertible);
such a column cannot be a multiple of a standard basis vector
nor the zero vector.
Similarly, if $y=0$, then
$\vec{\theta}=x\gTranspose{1pt}{\begin{bmatrix}u,0,-u+v,0,v,-u+v,-u+v\end{bmatrix}}$;
at least one of rows $1$ or $5$ has to be zero, thus,
w.l.o.g. suppose $v=0$, we obtain that
$\vec{\theta}=ux\gTranspose{1pt}{\begin{bmatrix}1,0,-1,0,0,-1,-1\end{bmatrix}}$;
such a column cannot be a multiple of a standard basis vector
nor the zero vector.
Finally, if $x=v=0$, then
$\vec{\theta}=uy\gTranspose{1pt}{\begin{bmatrix}0,1,-1,0,0,0,0\end{bmatrix}}$;
again that column cannot be a multiple of a standard basis vector
nor the zero vector.
\end{proof}
Now we show that any in-place elementary algorithm requires at least $1$
extra operation to put back the input in its initial state.
\begin{lemma}\label{lem:plusone}
Let $\vec{a}\in\F^m$ and $\mat{\alpha}\in\F^{t{\times}m}$ with at
least one row which is neither the zero row, nor a standard basis vector.
Now suppose that, without any constraints in terms of temporary
registers, $k$ is the minimal number of elementary operations required
to compute $\mat{\alpha}\vec{a}$.
Then any algorithm computing all the $t$ values of
$\mat{\alpha}\vec{a}$, in-place of $\vec{a}$, requires at least $k+1$
elementary operations.
\end{lemma}
\begin{proof}
Consider an in-place algorithm realizing $\mat{\alpha}\vec{a}$ in $f$
operations.
Any zero or standard basis vector row can
be realized without any operations on $\vec{a}$.
Now take this algorithm at the moment where the last of the other rows of
$\mat{\alpha}$ are realized (at that point all the $t$ values are
realized). Then this last realization (a non-trivial linear combination of the
initial values of $\vec{a}$) has to have been stored in one variable
of $\vec{a}$, say $a_i$.
Therefore, at this point, the in-place algorithm
has to perform at least one more operation to put back $a_i$ to its
initial state.
Therefore, by replacing all the in-place computations by operations on
extra registers and omitting the operation(s) that restore this $a_i$,
we obtain an algorithm with less than $f-1$ elementary operations that
realizes $\mat{\alpha}\vec{a}$ and thus: $(f-1)\geq{k}$.
\end{proof}

\begin{proposition}\label{prop:six}
For the in-place realization of each of the two linear operators
$\mat{\alpha}$ and $\mat{\beta}$, of any bilinear matrix
multiplication algorithm using only $7$ multiplications,
and the restoration of the initial states of their input,
at least $6$ operations are needed.
\end{proposition}
\begin{proof}
A bilinear matrix multiplication algorithm has to compute
$\mat{\alpha}\vec{a}$, with $\vec{a}$ the entries of the left input of
the matrix multiplication, while $\mat{\beta}$ deals with the right
input.
These $\mat{\alpha}$ and $\mat{\beta}$ matrices cannot contain a
($4$-dimensional) zero row:
otherwise there would exist an algorithm using less than $6$
multiplications, but $7$ is minimal~\cite{Winograd:1971:minseven}.
If $\mat{\alpha}$ or $\mat{\beta}$ contain at least $5$ rows that are not
standard basis vectors, then they require at
least $5$ non-trivial operations to be computed, and therefore at
least $6$ elementary operations with an in-place algorithm,
by~\cref{lem:plusone}.
The matrices also cannot contain more than $3$ multiples of standard
basis vectors, by~\cite[Lemma~8]{Bshouty:1995:minwinoadd}.
There thus remains now only to consider matrices with exactly $3$ rows
that are multiple of standard basis vectors.
Let $\mat{M}$ be the $4{\times}4$ sub-matrix obtained from
$\mat{\alpha}$ (or $\mat{\beta}$) by removing those $3$ standard basis vectors.
By~\cref{lem:nocannorzero}, no column of $\mat{M}$ can be the zero column:
otherwise a $7$-dimensional column of $\mat{\alpha}$ (or
$\mat{\beta}$) would be either a multiple of a standard basis vector,
or the zero vector.
This means that every variable of $\vec{a}$ has to be used at least
once to realize the $4$ operations of $\mat{M}\vec{a}$.
Now suppose that there exists an in-place algorithm realizing
$\mat{M}\vec{a}$ in $5$ elementary operations.
Any operations among these $5$ that, as its results, puts back a
variable into its initial state, does not realize any row of
$\mat{M}\vec{a}$ (because putting back a variable to its initial state
is the trivial identity on this initial variable, and this would be
represented by a $4$-dimensional standard basis vector, which
$M$ do not contain, by construction).
Therefore, at most one among these $5$ operations puts back a variable
of $\vec{a}$ into its initial state (otherwise $\mat{M}\vec{a}$, and
therefore  $\mat{\alpha}\vec{a}$ or $\mat{\beta}\vec{a}$, would be
realizable in strictly less than $4$ operations).
Thus, at most one variable of $\vec{a}$ can be modified during
the algorithm (otherwise the algorithm would not be able to put back
all its input variables into their initial state).

W.l.o.g suppose this only modified variable is $a_1$.
Finally, as all the other $3$ variables must be used in at least one
of the $5$ elementary operations, at least $3$ operations are
of the form $a_1\pe\lambda_i{a_i}$ for $i=2,3,4$ and some constants $\lambda_i$.
After those, to put back $a_1$ into its initial state, each one
of these $3$ independent variables, $a_2$, $a_3$ and $a_4$, must be ``removed''
from $a_1$ at some point of the elementary program.
But, with a total of $5$ operations, there remains only $2$ other possible
elementary operations, each one of those modifying only~$a_1$.
Therefore not all $3$ variables can be removed and thus no in-place
algorithm can use only $5$ operations.
\end{proof}
Finally, there remains to consider the linear combinations of the $7$
multiplications to conclude that~\cref{alg:ipsw} realizes the minimal
number of operations for any in-place algorithm with $7$
multiplications.
\begin{theorem}\label{thm:eighteen}
  At least $25$ additions are required to compute in-place
  any bilinear matrix multiplication algorithm using
  only $7$ multiplications and to restore
  its input matrices to their initial states afterwards.
\end{theorem}
\begin{proof} \Cref{prop:six} shows that at least $6$ operations are
  required to realize $\alpha$ (or $\beta$).
  For $\mu$, we in fact compute $\vec{c}\pe\mu\vec{\rho}$,
  so we need to consider the matrix
  $P=\begin{smatrix}I_4&\mu\end{smatrix}\in\F^{4{\times}11}$
  and
  the vector
  $\vec{\xi}=\begin{smatrix}\vec{c}\\\vec{\rho}\end{smatrix}$.
  Consider now an elementary program that realizes $P\vec{\xi}$,
  in-place of $\vec{c}$ only. This implies for instance that if
  $\vec{\rho}$ is zero, $\vec{c}$ should ultimately be put back to its initial
  state.
  Finally, consider the transposed program
  $\Transpose{P}\vec{\underline{c}}$: it must be in-place of
  $\vec{\underline{c}}$, while putting back $\vec{\underline{c}}$ to
  its initial state afterwards.
  By \cref{prop:six}, $\Transpose{\mu}$, thus
  $\Transpose{P}\in\F^{11{\times}4}$, requires at least $6$ elementary
  operations to be performed.
  By Tellegen's transposition principle, see
  also~\cite[Theorem~7]{Kaminski:1988:transpose}, computing the
  transposed program requires at least $6+(11-4)=13$ operations.
  This gives a total of at least $6+6+13=25$ additions.
\end{proof}

\Cref{thm:eighteen} thus shows that our~\cref{alg:ipsw} with $18$
elementary additions and $7$ recursive calls (thus $7$ more, and a
total of $18+7=25$ additions) is an optimal in-place
bilinear matrix multiplication algorithm using
only $7$ multiplications.

To go beyond our minimality result for operations, one could
try an alternate basis of~\cite{Karstadt:JACM:2020:MMfaster}.
But an argument similar to that of~\cref{prop:six} shows
that alternate basis does not help for the in-place case.

Any bilinear algorithm for matrix multiplication (see, e.g.,
\url{https://fmm.univ-lille.fr/}) can be dealt with similarly.
Further, even the accumulating version of the non-bilinear algorithm
of~\cite{Dumas:2023:adjoint} can
benefit from our techniques of in-place
accumulation.

%
%
\subsection{In-place Square \& Rank-k Update}\label{app:aat}
Thanks to~\cref{alg:ipsw} and with some care on transposes, the same
technique can be adapted to, e.g., \cite[Alg.~12]{Dumas:2023:adjoint},
which performs the multiplication of a matrix by its transpose.
With an accumulation, this is a classical \emph{Symmetric Rank-k
  Update} (or SYRK): $C\leftarrow{\alpha{A}\Transpose{A}+\beta{C}}$.

Following the notations of the latter algorithm, which is not a
bilinear algorithm on its single input matrix, the in-place
accumulating version is shown in~\cref{alg:ipaat}, for $\alpha=\beta=1$,
using any (fast to apply) skew-unitary $Y\in\F^{n{\times}n}$.
It has been obtained automatically by the method
of~\cref{thm:general}, and it thus preserves the need of only $5$
multiplications $P_1$ to ${P_5}$. It has then been scheduled to
reduce the number of extra operations.

\Cref{alg:ipaat} requires
$3$ recursive calls,
$2$ multiplications of two independent half matrices,
$4$ multiplications by a skew-unitary half matrix,
$8$ additions (of half inputs),
$12$ semi-additions (of half triangular outputs).
Provided that the multiplication by the skew-unitary matrix can be
performed in-place in negligible time,
this gives a dominant term of the complexity bound
for~\cref{alg:ipaat} of a fraction $\frac{2}{2^\omega-3}$ of the
cost of the full in-place algorithm.
This is a factor $\frac{1}{2}$, when~\cref{alg:ipsw} is used for the
two block multiplications of independent matrices ($P4$ and $P5$).

\begin{algorithm}[ht]
\caption{In-place accumulating multiplication by its transpose.}\label{alg:ipaat}
\begin{algorithmic}
\REQUIRE $A=\begin{smatrix} a_{11} & a_{12} \\  a_{21}
  &a_{22}\end{smatrix}\in\F^{m{\times}2n}$;
symmetric $C=\begin{smatrix} c_{11} & \Transpose{c_{21}} \\  c_{21}
  &c_{22}\end{smatrix}\in\F^{m{\times}m}$.
\ENSURE $\Low{C}\pe\Low{A\cdot\Transpose{A}}$. \hfill\COMMENT{update
  bottom left triangle}
\end{algorithmic}
\begin{Verbatim}[commandchars=\\\{\},codes={\catcode`$=3\catcode`^=7\catcode`_=8}]
Low(C$_{22}$)$:=\,$Low(C$_{22}$)${-}$Low(C$_{11}$);   Low(C$_{21}$)$:=\,$Low(C$_{21}$)${-}$Low(C$_{11}$);
 Up(C$_{21}$)$:=\,$ Up(C$_{21}$)${-}$Low(C$_{11}$)$^{\intercal}$;
\colorbox{cyan!10}{Low(C$_{11}$)$:=\,$Low(C$_{11}$)${+}$Low$(A_{11}${*}A$_{11}^{\intercal}$);}       # $P_1$ Rec.
 Up(C$_{21}$)$:=\,$ Up(C$_{21}$)${+}$Low(C$_{11}$)$^{\intercal}$;
Low(C$_{21}$)$:=\,$Low(C$_{21}$)${+}$Low(C$_{11}$);   Low(C$_{22}$)$:=\,$Low(C$_{22}$)${+}$Low(C$_{11}$);
\colorbox{cyan!10}{Low(C$_{11}$)$:=\,$Low(C$_{11}$)${+}$Low$(A_{12}{*}$A$_{12}^{\intercal}$);}       # $P_2$ Rec.
A$_{11}:=\,$A$_{11}{*}$Y;   A$_{21}:=\,$A$_{21}{*}$Y; A$_{11}:=\,$A$_{11}{-}$A$_{21}$; A$_{21}:=\,$A$_{21}{-}$A$_{22}$;
Low(C$_{22}$)$:=\,$Low(C$_{22}$)${-}$Low(C$_{21}$);   Low(C$_{22}$)$:=\,$Low(C$_{22}$)${-}$Low(C$_{21}^{\intercal}$);
\colorbox{cyan!10}{C$_{21}:=\,$C$_{21}{+}$A$_{11}{*}$A$_{21}^{\intercal}$;}                      # $P_4$ ({e.g.}, \cref{alg:ipsw})
Low(C$_{22}$)$:=\,$Low(C$_{22}$)${+}$Low(C$_{21}^{\intercal}$);
A$_{21}:=\,$A$_{21}{-}$A$_{11}$;
 Up(C$_{21}$)$:=\,$ Up(C$_{21}$)${-}$Low(C$_{21}$)$^{\intercal}$;
\colorbox{cyan!10}{Low(C$_{21}$)$:=\,$Low(C$_{21}$)${+}$Low$(A_{21}{*}$A$_{21}^{\intercal}$);}       # $P_5$ Rec.
 Up(C$_{21}$)$:=\,$ Up(C$_{21}$)${+}$Low(C$_{21}$)$^{\intercal}$;   Low(C$_{22}$)$:=\,$Low(C$_{22}$)${+}$Low(C$_{21}$);
A$_{21}:=\,$A$_{21}{+}$A$_{12}$;
\colorbox{cyan!10}{C$_{21}:=\,$C$_{21}{+}$A$_{22}{*}$A$_{21}^{\intercal}$;}                      # $P_3$ ({e.g.}, \cref{alg:ipsw})
A$_{21}:=\,$A$_{21}{-}$A$_{12}$; A$_{21}:=\,$A$_{21}{+}$A$_{11}$; A$_{21}:=\,$A$_{21}{+}$A$_{22}$; A$_{11}:=\,$A$_{11}{+}$A$_{21}$;
A$_{21}:=\,$A$_{21}{*}$Y$^{-1}$; A$_{11}:=\,$A$_{11}{*}$Y$^{-1}$;
\end{Verbatim}
\end{algorithm}

Now, the skew-unitary matrices used in~\cite{Dumas:2023:adjoint}, are
either a multiple of the identify matrix, or the Kronecker product of
$\begin{smatrix}a&b\\-b&a\end{smatrix}$ by the identity matrix,
for $a^2+b^2=-1$ and $a\neq{0}$.
The former is easily performed in-place in time \bigO{n^2}.
For the latter, it is sufficient to use~\cref{eq:twobytwomul}:
the multiplication $\begin{smatrix}a&b\\-b&a\end{smatrix}\vec{u}$
can be realized in place by the algorithm:
$u_1\fe{a}$; $u_1\pe{b{\cdot}u_2}$; $u_2\fe(a+b^2a^{-1})$;
$u_2\pe{\left(-ba^{-1}\right){\cdot}u_1}$.
The same technique can be used on the symmetric algorithm for the
square of matrices given in~\cite{Bodrato:ISSAC2010}. The resulting
in-place algorithm is given in~\cref{alg:square}.
\begin{algorithm}[htbp]
\caption{In-place accumulating S.-W. matrix-square.}\label{alg:square}
\begin{algorithmic}
\REQUIRE
$A=\begin{smatrix} a_{11} & a_{12} \\  a_{21} &a_{22}\end{smatrix}$,
$C=\begin{smatrix} c_{11} & c_{12} \\  c_{21} &c_{22}\end{smatrix}$.
\ENSURE $C\pe{A^2}$.
\end{algorithmic}
\begin{Verbatim}[commandchars=\\\{\},codes={\catcode`$=3\catcode`^=7\catcode`_=8}]
A$_{22}$ := A$_{22}$ - A$_{21}$; C$_{12}$ := C$_{12}$ + C$_{22}$;
\colorbox{cyan!10}{C$_{22}$ := C$_{22}$ + A$_{22}$ * A$_{22}$;}
A$_{22}$ := A$_{22}$ + A$_{12}$; A$_{22}$ := A$_{22}$ - A$_{11}$;
\colorbox{cyan!10}{C$_{12}$ := C$_{12}$ - A$_{22}$ * A$_{12}$;}
\colorbox{cyan!10}{C$_{21}$ := C$_{21}$ - A$_{21}$ * A$_{22}$;}
C$_{21}$ := C$_{21}$ - C$_{22}$; A$_{22}$ := A$_{22}$ + A$_{11}$;
\colorbox{cyan!10}{C$_{22}$ := C$_{22}$ - A$_{22}$ * A$_{22}$;}
C$_{11}$ := C$_{11}$ + C$_{22}$;
\colorbox{cyan!10}{C$_{22}$ := C$_{22}$ - A$_{12}$ * A$_{21}$;}
A$_{22}$ := A$_{22}$ + A$_{21}$; C$_{12}$ := C$_{12}$ - C$_{22}$; C$_{11}$ := C$_{11}$ - C$_{22}$;
\colorbox{cyan!10}{C$_{22}$ := C$_{22}$ + A$_{22}$ * A$_{22}$;}
A$_{22}$ := A$_{22}$ - A$_{12}$; C$_{21}$ := C$_{21}$ + C$_{22}$;
\colorbox{cyan!10}{C$_{11}$ := C$_{11}$ + A$_{11}$ * A$_{11}$;}
\end{Verbatim}
\end{algorithm}

\section{In-place polynomial multiplication with
  accumulation}\label{sec:inpaccpol}
\Cref{alg:bilin} can also be used for polynomial multiplication.
An additional difficulty is that this does not completely fit the setting,
as multiplication of two size-$n$ inputs will in general span a (double)
size-$2n$ output.
This is not an issue until one has to distribute separately the two
halves of this $2n$ values (or more generally to different parts of
different outputs).
In the following we show that this can anyway always be done
for polynomial multiplications.

\subsection{In-place accumulating Karatsuba}\label{ssec:kara}
For instance, we immediately obtain an in-place Karatsuba
polynomial multiplication since it writes as
in~\cref{eq:kara}, from which we can extract the associated $\mu$,
$\alpha$, $\beta$ matrices shown in~\cref{eq:bilinkara}.
\begin{gather}
\begin{aligned}
(Ya_1 + a_0)(Yb_1 +b_0)&=a_0b_0+Y^2(a_1b_1)\\
&+Y(a_0b_0+a_1b_1-(a_0-a_1)(b_0-b_1))
\end{aligned}\label{eq:kara}\\
\mu=\begin{smatrix}
1&0&0\\
1&1&-1\\
0&1&0
\end{smatrix}\quad\quad
\alpha=\begin{smatrix}
1&0\\
0&1\\
1&-1
\end{smatrix}\quad\quad
\beta=\begin{smatrix}
1&0\\
0&1\\
1&-1
\end{smatrix}\label{eq:bilinkara}
\end{gather}

Then, with $Y=X^\delta$ and $a_i$, $b_i$, $c_i$ polynomials in $X$ (and $a_0$,
$b_0$, $c_0$ of degree less than $t$), this is
detailed, with accumulation, in~\cref{eq:karasplit}:
\begin{equation}\label{eq:karasplit}
\fbox{\scalebox{.975}[0.975]{\ensuremath{\begin{aligned}
A(Y)& = Ya_1 + a_0;\quad
B(Y) = Yb_1 + b_0;\\
C(Y)& = Y^3c_{11} + Y^2c_{10} + Yc_{01} + c_{00};\\
m_0& = a_0\cdot{b_0} = m_{01}Y+m_{00};\quad
m_1 = a_1\cdot{b_1} = m_{11}Y+m_{10};\\
m_2& = (a_0 - a_1)\cdot(b_0 - b_1)= m_{21}Y+m_{20};\\
t_{00} &= c_{00}+m_{00};\quad
t_{01}  = c_{01}+m_{01}+m_{00}+m_{10}-m_{20};\\
t_{10} &= c_{10}+m_{10}+m_{01}+m_{11}-m_{21};\quad
t_{11} = c_{11}+m_{11};\\
\text{\algorithmicthen}&\quad C+AB ={Y^3t_{11}+Y^2t_{10}+Yt_{01}+t_{00}}
\end{aligned}}}}
\end{equation}
To deal with the distributions of each half of the
products of~\cref{eq:karasplit},
each coefficient in $\mu$ in~\cref{eq:bilinkara} can be expanded into
$2{\times}2$ identity blocks, and the middle rows combined two by two,
as each tensor product actually spans
two sub-parts of the result; we obtain~\cref{eq:bilinkarasplit}:
\begin{equation}\label{eq:bilinkarasplit}
\scalebox{.95}[0.95]{%
\ensuremath{\mu^{(2)}{=}
\begin{smatrix} I_2&0_2&0_2\\0_2&0_2&0_2\end{smatrix}
{+}
\begin{smatrix} 0&0&0\\I_2&I_2&-I_2\\ 0&0&0\end{smatrix}
{+}
\begin{smatrix} 0_2&0_2&0_2\\0_2&I_2&0_2\end{smatrix}
{=}
\begin{smatrix}
1&0&0&0&0&0\\
1&1&1&0&-1&0\\
0&1&1&1&0&-1\\
0&0&0&1&0&0
\end{smatrix}}.}
\end{equation}
Finally, \cref{eq:karasplit} then translates into an in-place algorithm
thanks to~\cref{alg:bilin,eq:bilinkara,eq:bilinkarasplit}.
The first point is that
products double the degree: This corresponds to a constraint that the
two blocks have to remain together when distributed.
In other words, this means that the matrix $\mu^{(2)}$ needs to be
considered two consecutive columns by two consecutive columns.
This is always possible if the two columns are of full rank $2$.
Indeed, consider a $2\times{2}$ invertible sub-matrix
$M=\begin{smatrix} v & w\\x & y\end{smatrix}$ of these two columns.
Then computing
$\begin{smatrix}c_i\\c_j\end{smatrix}\pe{}M\begin{smatrix}\rho_0\\\rho_1\end{smatrix}$
is equivalent to computing a $2\times{2}$ version of~\cref{eq:basemul}:
\begin{equation}\label{eq:twobytwo}
\left\lbrace\begin{smatrix}c_i\\c_j\end{smatrix} \fe M^{-1}; \quad
\begin{smatrix}c_i\\c_j\end{smatrix} \pe \begin{smatrix} \rho_0\\\rho_1\end{smatrix}; \quad
\begin{smatrix}c_i\\c_j\end{smatrix} \fe M\right\rbrace.
\end{equation}
The other rows of these two columns can be dealt with as before by
pre- and post-multiplying/dividing by a constant and pre- and
post-adding/subtracting the adequate $c_i$ and $c_j$.
Now to apply a matrix $M=\begin{smatrix}a&b\\c&d\end{smatrix}$ to a
vector of results $\begin{smatrix}\vec{u}\\\vec{v}\end{smatrix}$,
it is sufficient that one of its coefficients is invertible.
W.l.o.g suppose that its upper left element, $a$, is invertible.
Thus,
$\begin{smatrix}a&b\\c&d\end{smatrix}=\begin{smatrix}1&0\\ca^{-1}&1\end{smatrix}\begin{smatrix}a&b\\0&d-ca^{-1}\end{smatrix}$.
Then the in-place evaluation of~\cref{eq:twobytwomul} performs this
application, using the two (known in advance) constants $x=ca^{-1}$ and $y=d-ca^{-1}b$:
\begin{equation}\label{eq:twobytwomul}
  \fbox{\ensuremath{\left.\begin{aligned}
	  \vec{u}&\fe{a}\\[-5pt]
	  \vec{u}&\pe{b\cdot\vec{v}}\\[-5pt]
	  \vec{v}&\fe{y}\\[-5pt]
	  \vec{v}&\pe{x\cdot\vec{u}}
   \end{aligned}\right\rbrace
   \quad
   \begin{array}{l}
	\text{computes in-place:}\\
	\begin{smatrix}\vec{u}\\\vec{v}\end{smatrix}
	\gets\begin{smatrix}a&b\\c&d\end{smatrix}
	\odot
	\begin{smatrix}\vec{u}\\\vec{v}\end{smatrix}
	=\begin{smatrix}a\vec{u}+b\vec{v}\\c\vec{u}+d\vec{v}\end{smatrix}\\
	\text{for}~x=ca^{-1}~\text{and}~y=d-xb
      \end{array}
    }}
\end{equation}

\begin{remark}\label{rq:zerotopleft}
  In practice for $2\times{2}$ blocks, if $a$ is not
  invertible, permuting the rows is sufficient since $c$ has to be
  invertible for the matrix to be invertible: for
  $J=\begin{smatrix}0&1\\1&0\end{smatrix}$,
  if $\tilde{M}=\begin{smatrix}c&d\\0&b\end{smatrix}=J{\cdot}{M}$, then
  $M=J{\cdot}\tilde{M}$ and $M^{-1}=\tilde{M}^{-1}{\cdot}{J}$ so
  that~\cref{eq:twobytwo} just becomes:\\
\(\begin{smatrix} c_i\\c_j\end{smatrix}\fe J;
  \begin{smatrix} c_i\\c_j\end{smatrix}\fe \tilde{M}^{-1};
  \begin{smatrix} c_i\\c_j\end{smatrix}\pe\begin{smatrix}\rho_0\\\rho_1\end{smatrix};
  \begin{smatrix} c_i\\c_j\end{smatrix}\fe \tilde{M};
  \begin{smatrix} c_i\\c_j\end{smatrix}\fe J.\)
\end{remark}

We now have the tools for in-place polynomial algorithms.
We start, in~\cref{alg:doublebilin}, with a version
of~\cref{alg:bilin} for which the multiplications are accumulated into
two consecutive blocks (denoted \MUL-2D).
\begin{algorithm}[htbp]
  \caption{In-place bilinear $2$ by $2$ formula.}\label{alg:doublebilin}
  \begin{algorithmic}[1]\small
    \REQUIRE $\vec{a}\in\F^m$, $\vec{b}\in\F^n$, $\vec{c}\in\F^s$;
    $\mat{\alpha}\in\F^{t{\times}m}$, $\mat{\beta}\in\F^{t{\times}n}$,
    $\mat{\mu}\in\F^{s{\times}(2t)}=\begin{smatrix}M_1&\cdots&M_t\end{smatrix}$,
      with no zero-rows in $\alpha$, $\beta$, $\mu$,
      s.t. $(a_i\cdot{b_j})$ fits two result variables $c_k$, $c_l$
      and s.t. $M_i\in\F^{s{\times}2}$ is of full-rank $2$ for $i=1..t$.
    \READONLY$\mat{\alpha},\mat{\beta},\mat{\mu}$.
    \ENSURE $\vec{c}\pe\mat{\mu}\vec{m}$, for
    $\vec{m}=(\mat{\alpha}\vec{a})\odot(\mat{\beta}\vec{b})$
    \FOR{$\ell=1$ \To $t$}
    \STATE Let $i$ s.t. $\alpha_{\ell,i}\neq{0}$;
    $a_i\fe\alpha_{\ell,i}$;
    \ForDoEnd[lin:doublealpha]{$\lambda=1$ \To $m$, $\lambda\neq{i}$,
      $\alpha_{\ell,\lambda}\neq{0}$
    }{$a_i\pe\alpha_{\ell,\lambda}a_\lambda$}
    \STATE Let $j$ s.t. $\beta_{\ell,j}\neq{0}$;
    $b_j\fe\beta_{\ell,j}$;
    \ForDoEnd[lin:doublebeta]{$\lambda=1$ \To $n$, $\lambda\neq{j}$,
      $\beta_{\ell,\lambda}\neq{0}$}{$b_j\pe\beta_{\ell,\lambda}b_\lambda$}
    \STATE Let $k$, $f$
    s.t. $M=\begin{smatrix}\mu_{k,2\ell}&\mu_{k,2\ell+1}\\\mu_{f,2\ell}&\mu_{f,2\ell+1}\end{smatrix}$
    is invertible;
    \STATE\label{lin:invmul}$\begin{smatrix} c_k\\c_f \end{smatrix} \leftarrow M^{-1}
    \begin{smatrix} c_k\\c_f \end{smatrix}$
    \hfill\COMMENT{via~\cref{eq:twobytwomul,rq:zerotopleft}}
    \ForDoEnd[lin:divsube]{$\lambda=1$ \To $s$,
      $\lambda\not\in\{f,k\}$,
      $\mu_{\lambda,2\ell}\neq{0}$}{$c_\lambda\me\mu_{\lambda,2\ell}{c_k}$}
    \ForDoEnd[lin:divsubo]{$\lambda=1$ \To $s$,
      $\lambda\not\in\{f,k\}$,
      $\mu_{\lambda,2\ell+1}\neq{0}$}{$c_\lambda\me\mu_{\lambda,2\ell+1}{c_f}$}
    \STATE\label{lin:doubleproduct}$\begin{smatrix} c_k\\c_f
    \end{smatrix}\pe{a_i\cdot{b_j}}$\hfill\COMMENT{This is the
      accumulation of the product $\begin{smatrix} m_k\\m_f\end{smatrix}$}
    \ForDoEnd{$\lambda=1$ \To $s$, $\lambda\not\in\{f,k\}$,
      $\mu_{\lambda,2\ell+1}\neq{0}$}{$c_\lambda\pe\mu_{\lambda,2\ell+1}{c_f}$}
    \ForDoEnd{$\lambda=1$ \To $s$,$\lambda\not\in\{f,k\}$,
      $\mu_{\lambda,2\ell}\neq{0}$}{$c_\lambda\pe\mu_{\lambda,2\ell}{c_k}$}
    \STATE$\begin{smatrix} c_k\\c_f \end{smatrix} \leftarrow M\begin{smatrix} c_k\\c_f \end{smatrix}$
    \hfill\COMMENT{via~\cref{eq:twobytwomul,rq:zerotopleft}, undo~\ref{lin:invmul}}
    \ForDoEnd{$\lambda=1$ \To $n$, $\lambda\neq{j}$,
      $\beta_{\ell,\lambda}\neq{0}$}{$b_j\me\beta_{\ell,\lambda}b_\lambda$};~$b_j\de\beta_{\ell,j}$;
    \ForDoEnd{$\lambda=1$ \To $m$, $\lambda\neq{i}$,
      $\alpha_{\ell,\lambda}\neq{0}$}{$a_i\me\alpha_{\ell,\lambda}a_\lambda$};~$a_i\de\alpha_{\ell,i}$;
    \ENDFOR
    \RETURN $\vec{c}$.
  \end{algorithmic}
\end{algorithm}
A C++ implementation of~\cref{alg:bilin,alg:doublebilin}
(\href{https://github.com/jgdumas/plinopt/blob/main/src/trilplacer.cpp}{\texttt{trilplacer}})
is available in the
\href{https://github.com/jgdumas/plinopt}{\textsc{PLinOpt}}
 library:~\url{https://github.com/jgdumas/plinopt}.
\begin{theorem}\label{thm:doublebilin}
\Cref{alg:doublebilin} is correct, in-place, and requires
$t$ \MUL-2D,
$2(\#\alpha+\#\beta+\#\mu-t)$ \ADD and
$2(\sharp\alpha+\sharp\beta+\sharp\mu+2t)$ \SCA operations.
\end{theorem}
\begin{proof}
  Thanks to~\cref{eq:twobytwo,eq:twobytwomul,rq:zerotopleft},
  correctness is similar to that of~\cref{alg:bilin}
  in~\cref{thm:bilin}.
Then, \cref{eq:twobytwomul} requires $4$ \SCA
and $2$ \ADD operations and is called $2t$ times.
The rest is similar to~\cref{alg:bilin} and amounts to
$2t+2(\#\alpha-t+\#\beta-t+\#\mu-2t)+(2t)2$ \ADD and
$2(\sharp\alpha+\sharp\beta+\sharp\mu-2t)+(2t)4$ \SCA operations.
\end{proof}

There remains to use a double expansion of the output
$\mu\in\F^{s{\times}t}$ to simulate the double size of the
intermediate products (\MUL-2D), producing
$\mu^{(2)}\in\F^{s{\times}(2t)}$, as in~\cref{eq:bilinkarasplit}, that
is used as an input in~\cref{alg:doublebilin}.
This double expansion matrix is obtained by the following duplication:
$\mu^{(2)}(i,2j)=\mu(i,j)$ and $\mu^{(2)}(i+1,2j+1)=\mu(i,j)$ for
$i=1..s$ and $j=1..t$. We prove, in \cref{lem:fullrank},
that in fact any such double expansion of a representative matrix
is suitable for the in-place computation of~\cref{alg:doublebilin}.

\begin{lemma}\label{lem:fullrank}
  If $\mu$ does not contain any zero column, then each pair of columns
  of $\mu^{(2)}$, resulting from the expansion of a single column in $\mu$,
  contains an invertible lower triangular $2{\times}2$ sub-matrix.
\end{lemma}
\begin{proof}
The top most non-zero element of a column
is expanded as a $2{\times}2$ identity matrix whose second row is
merged with the first row of the next identity matrix: %
$\begin{smatrix} a \\ b\end{smatrix}$ is expanded to
$\begin{smatrix} a &0 \\ b & a\\ * & b\end{smatrix}$.
\end{proof}

For instance with $m_{00}+Ym_{01}=a_0b_0=\rho_0+Y\rho_1$,
consider the upper left $2\times{2}$ block of $\mu^{(2)}$
in~\cref{eq:bilinkarasplit}, that is
$M=\begin{smatrix} 1&0\\1&1\end{smatrix}$, whose inverse is
$M^{-1}=\begin{smatrix} 1&0\\-1&1\end{smatrix}$.
One has first to precompute
$M^{-1}\begin{smatrix}c_{00}\\c_{01}\end{smatrix}$, that is nothing to
$c_{00}$ and $c_{01}\me{}c_{00}$ for the second coefficient.
Then, afterwards, the third row, for $c_{10}$, will just be $-m_{01}$:
for this just pre-subtract $c_{10}\me{}c_{01}$, and post-add
$c_{10}\pe{}c_{01}$ after the product actual computation.
This example is in lines \ref{lin:m0BEG} and \ref{lin:m0END}
of~\cref{alg:accinplmulkara} thereafter.
To complete~\cref{eq:karasplit}, the computation of $m_1$ is dealt
with in the same manner, while that
of $m_2$ is direct in the results.
Note that as both $t_{01}$ and $t_{10}$ receive together
$m_{01}+m_{10}$, some pre- and post-additions are simplified out
in~\cref{alg:accinplmulkara}.
The second point is to deal with unbalanced dimensions and degrees
for $Y=X^\delta$ and recursive calls. First split the largest
polynomial into two parts, so that two sub-products are performed: a
large balanced one, and, recursively, a smaller unbalanced one.
Second, for the balanced case, the idea is to ensure that three out of
four parts of the result, $t_{00}$, $t_{01}$ and $t_{10}$, have the
same size and that the last one $t_{11}$ is smaller. This ensures that
all accumulations can be performed in-place.
Details can be found in~\cref{alg:accinplmulkara}.

\begin{algorithm}[ht]
\caption{In-place Karatsuba polynomial multiplication with
  accumulation}\label{alg:accinplmulkara}
\begin{minipage}{\columnwidth}
\begin{algorithmic}[1]
\REQUIRE $A$, $B$, $C$ polynomials of degrees $m$, $n$, $m+n$
with $m\geq{n}$.
\ENSURE $C\pe{AB}$
\IF{$n\leq\threshold$}\hfill\COMMENT{Constant-time if $\threshold\in\bigO{1}$}
\RETURN the quadratic in-place multiplication. \hfill\COMMENT{\Cref{alg:classicmul}}
\ELSIF{$m>n$}
\STATE Let $A(X)=A_0(X)+X^{n+1}A_1(X)$
\STATE $C_{0..2n}\pe{A_0B}$\hfill\COMMENT{Recursive call}
\IF{$m\geq{2n}$}
\STATE $C_{(n+1)..(n+m)}\pe{A_1B}$\hfill\COMMENT{Recursive call}
\ELSE
\STATE $C_{(n+1)..(n+m)}\pe{BA_1}$\hfill\COMMENT{Recursive call}
\ENDIF
\ELSE\hfill\COMMENT{Now $m=n$}
\STATE Let $\delta=\lceil(2n+1)/4\rceil$; \hfill\COMMENT{$\delta-1\geq{2n-3\delta}$ and thus $\delta>n-\delta$}
\STATE Let $A=a_0+X^{\delta}a_1$; $B=b_0+X^{\delta}b_1$;
\STATE Let $C=c_{00}+c_{01}X^{\delta}+c_{10}X^{2\delta}+c_{11}X^{3\delta}$;
\hfill\COMMENT{$d^\circ{c_{11}}=2n-3\delta$}
\STATE\label{lin:m0BEG}$c_{01}\me{c_{00}}$;\quad $c_{10}\me{c_{01}}$;
\STATE $\begin{smatrix} c_{00} \\ c_{01}\end{smatrix} \pe a_0\cdot{b_0} $
\hfill\COMMENT{Recursive call for $m_0$}
\STATE $c_{11}\me{c_{10}{\scriptstyle[0..2n{-}3\delta]}}$;
\hfill\COMMENT{first $2n-3\delta+1$ coefficients of $c_{10}$}
\STATE $\begin{smatrix} c_{01} \\ c_{10}\end{smatrix} \pe a_1\cdot{b_1} $
\hfill\COMMENT{Recursive call for $m_1$}
\STATE $c_{11}\pe{c_{10}{\scriptstyle[0..2n{-}3\delta]}}$;
\hfill\COMMENT{as $d^\circ m_{11} \leq 2n-3\delta$}
\STATE\label{lin:m0END}$c_{10}\pe{c_{01}}$;\quad $c_{01}\pe{c_{00}}$;
\STATE $a_{0}\me{a_{1}}$;\quad\quad $b_{0}\me{b_{1}}$;\hfill\COMMENT{$d^\circ{a_0}=\delta-1\geq{n-\delta}=d^\circ{a_1}$}
\STATE $\begin{smatrix} c_{01} \\ c_{10}\end{smatrix} \me a_0\cdot{b_0}$
    \hfill\COMMENT{Recursive call\footnote{Variant that computes $C\me AB$, obtained by changing signs at each recursive call.} for $m_2$}
\STATE $b_{0}\pe{b_{1}}$;\quad\quad $a_{0}\pe{a_{1}}$;
\ENDIF
\RETURN $C$.
\end{algorithmic}
\end{minipage}
\end{algorithm}

\begin{proposition}
\cref{alg:accinplmulkara} is correct and requires
$\bigO{mn^{\log_2(3)-1}}$ operations.
\end{proposition}
\begin{proof}
With the above analysis, correctness comes from that
of~\cref{alg:doublebilin} applied
to~\cref{eq:bilinkara}.
When $m=n$, with $3$ recursive calls and \bigO{n} extra operations,
the algorithm thus requires overall $\bigO{n^{\log_2(3)}}$ operations.
Otherwise, it requires $\left\lfloor\frac{m}{n}\right\rfloor$ equal
degree calls, then a recursive call with $n$ and $(m\bmod{n})$.
Let $u_1=m$, $u_2=n$, $u_3$, \dots, $u_k$ denote the
successive residues in the Euclidean algorithm on inputs $m$ and $n$
(where $u_k$ is the last nonzero residue).
Then, \cref{alg:accinplmulkara} requires
less than
$\bigO{\sum_{i=1}^{k-1}\lfloor\frac{u_i}{u_{i+1}}\rfloor{}u_{i+1}^{\log_2(3)}}
\leq \bigO{\sum_{i=1}^{k-1} u_iu_{i+1}^{\log_2(3)-1}}$
operations.
But, $u_{i+1}\leq{u_2}=n$ and if we let $s_i=u_i+u_{i+1}$, 
$u_i\leq{s_i}$.
Thus, the number of operations is bounded by
$\bigO{\sum_{i=1}^{k-1} s_i n^{\log_2(3)-1}}$. 
From~\cite[Corollary~2.6]{Grenet:2020:euclide}, we have that 
$s_i\leq s_1(2/3)^{i-1}$. Therefore, $\sum_{i=1}^{k-1} s_i
\leq s_1 \sum_{i\geq0}(2/3)^i = \bigO{m+n}$,
and the number of operations is $\bigO{mn^{\log_2(3)-1}}$.
\end{proof}

Note that all coefficients of $\mat{\alpha}$,
$\mat{\beta}$ and $\mu^{(2)}$ being $1$ or $-1$, \cref{alg:accinplmulkara}
does compute the accumulation $C\pe{AB}$ without constant
multiplications.
Also, the de-duplication enables some
natural reuse.
There is thus a cost of
$2(\#\alpha-t+\#\beta-t)=2(4-3+4-3)=4$ additions with
$a_0,a_1,b_0,b_1$.
Then $2(2(\#\mu-t)-1)=2(\#\mu^{(2)}-2t-1)=2(10-6-1)$ additions with
$c_{ij}$ (\emph{i.e.}, $10-6$ minus the one saved by factoring
$m_{01}+m_{10}$).
This is a total of $3$ recursive
accumulating calls and at most $10$ half-block additions.
For degree $n-1$ (size $n$) polynomials, this is between $5n-\tfrac12$ and
$5n-5$ additional additions, and with $2$ operations for an
accumulating base case, this gives
a dominant term between $11.75n^{\log_2(3)}$ and $9.5n^{\log_2(3)}$.
A careful Karatsuba implementation for both polynomials of {\em size} $n$ a
power of two requires instead $6.5n^{\log_2(3)}$
operations~\cite{Roche:2009:spacetime}.
Now in practice, the supplementary operations are in fact, at least
mostly, compensated by the gain in memory allocations or movements:
depending on the compilator and flags we obtain some slight speed-up
or slow-down, about $\pm{10\%}$, when compared to the state of the art
Karatsuba modular ($60$ bits prime) polynomial multiplication of the
\href{https://github.com/libntl/ntl}{NTL} (\url{https://libntl.org})
library.

We compare in~\cref{tab:kara} the procedure given
in~\cref{alg:accinplmulkara} (obtained via the automatic application
of~\cref{alg:doublebilin}) with previous Karatsuba-like algorithms for
polynomial multiplications, designed to reduce their memory footprint
(see also \cite[Table~2.2]{Giorgi:2019:hdr}).

\begin{table}[ht]\centering
\caption{Reduced-memory algorithms for Karatsuba polynomial
  multiplication}\label{tab:kara}
\begin{tabular}{lcccc}
\toprule
\multirow{2}{*}{Algorithm} & \multicolumn{2}{c}{Memory Reg.} & \multirow{2}{*}{Inputs} & \multirow{2}{*}{Accumulation} \\
& Algebraic& Pointer & &\\
\midrule
\cite{Thome:2002:karatemp} & $n$ & $5\log{n}$ & {\color{teal} read-only} & {\xno}\\
\cite{Roche:2009:spacetime,Roche:2011:waterloo} & $0$ & $5\log{n}$ & {\color{teal} read-only} & {\xno}\\
\cite{Giorgi:2019:issac:reductions} &  \multicolumn{2}{c}{\textcolor{teal}{$\bigO{1}$}} & {\color{teal} read-only} & {\xno}\\
\cref{alg:accinplmulkara} & $0$ & $5\log n$ & mutable & {\xyes}\\
\bottomrule
\end{tabular}
\end{table}

\subsection{Further bilinear polynomial multiplications}\label{ssec:toom}

We have shown that any bilinear algorithm can be transformed into an
in-place version. This approach thus also works for any Toom-$k$
algorithm using $2k-1$ interpolations points instead of the three
points of Karatsuba (Toom-$2$).

For instance Toom-$3$ uses interpolations at
$0,1,-1,2,\infty$. Therefore, $\alpha$ and $\beta$ are the Vandermonde matrices
of these points for the $3$ parts of the input polynomials and
$\mu$ is the inverse of the Vandermonde matrix of these points for the
$5$ parts of the result, as shown in~\cref{eq:toom3} thereafter.

\begin{equation}\label{eq:toom3}
\mu=\begin{smatrix}
1 & 0 & 0 & 0 & 0 \\
1 & 1 & 1 & 1 & 1 \\
1 & -1 & 1 & -1 & 1 \\
1 & 2 & 4 & 8 & 16 \\
0 & 0 & 0 & 0 & 1
\end{smatrix}^{-1}\!\!\!\!=
\begin{smatrix}
 1 & 0 & 0 & 0 & 0 \\
{\scriptscriptstyle{-}}\frac{1}{2} & 1 & {\scriptscriptstyle{-}}\frac{1}{3} &  {\scriptscriptstyle{-}}\frac{1}{6} & 2 \\
 -1 & \frac{1}{2} & \frac{1}{2} & 0 & -1 \\
\frac{1}{2} & {\scriptscriptstyle{-}}\frac{1}{2} & {\scriptscriptstyle{-}}\frac{1}{6} & \frac{1}{6} & -2 \\
 0 & 0 & 0 & 0 & 1
\end{smatrix}\!;~\alpha=\beta=\begin{smatrix}
1 & 0 & 0 \\
1 & 1 & 1 \\
1 & -1 & 1 \\
1 & 2 & 4 \\
0 & 0 & 1
\end{smatrix}
\end{equation}

With the same kind of duplication as in~\cref{eq:bilinkarasplit}, apart from the
recursive calls, the initially obtained operation count is
$2(11+11-2*5)+2(2(16-5))=68$ additions
and $2(2+2+2(11))=52$ scalar multiplications.
Following the optimization of \cite{Bodrato:2007:WAIFI:toomcook},
we see in $\alpha$ and $\beta$ that the evaluations at $1$ and $-1$
(second and third rows) share one addition. As they are successive in
our main loop, subtracting one at the end of the second iteration,
then followed by re-adding it at the third iteration can be optimized
out. This is two fewer operations.
Together with shared coefficients in the rows of $\mu$,
some further optimizations of \cite{Bodrato:2007:WAIFI:toomcook} can
probably also be applied, where the same multiplicative
constants appear at successive places.
Currently, for instance,
\href{https://github.com/jgdumas/plinopt}{\textsc{PLinOpt}}
produces a program with only $56$ additions and $51$ scalar multiplications.

\subsection{Fast bilinear polynomial multiplication}\label{ssec:fft}

When sufficiently large roots of unity exist, polynomial
multiplications can be computed fast in our in-place model via a
discrete Fourier transform and its inverse. %
For simplicity,
we consider a ring $\D$ that contains a principal $N$-th root of unity $\omega\in\D$
for some $N = 2^p$. (In particular, $2$ is a unit in $\D$.)

Let $F\in\D[X]$ of degree $<N$.
The discrete Fourier transform of $F$ at $\omega$ is defined as
$\DFT_{N}(F,\omega) = (F(\omega^0), F(\omega^1), \dotsc, F(\omega^{N-1}))$.
The map is invertible, of inverse $\DFT^{-1}_{N}(\cdot,\omega) = \frac{1}{N} \DFT_{N}(\cdot,
\omega^{-1})$. Further, the DFT can be computed over-place, replacing the input by the
output~\cite{1965:CooleyTukey:MathComp:FFT}. Actually, for over-place
algorithms and their extensions to the \emph{truncated Fourier transform}, it is
more natural to work with the \emph{bit-reversed DFT} defined by
$\brDFT_{N}(F,\omega) = (F(\omega^{[0]_{p}}), F(\omega^{[1]_{p}}), \dotsc, F(\omega^{[N-1]_{p}}))$
where $[i]_{p}=\sum_{j=0}^{p-1} d_j 2^{p-j}$ is the length-$p$ bit reversal of $i = \sum_{j=0}^{p-1} d_j 2^j$,
$d_j\in\{0,1\}$.
Its inverse is $\brDFT_N^{-1}(\Lambda,\omega) = \frac{1}{N}\DFT_N((\Lambda_{[0]_p},\dots,\Lambda_{[N-1]_p}),\omega^{-1})$.

\begin{remark}
    The Fast Fourier Transform (FFT) algorithm has two main variants:
    \emph{decimation in time} (DIT) and \emph{decimation in frequency} (DIF).
    Both algorithms can be performed over-place, replacing the input by the
    output. Without applying any permutation to the entries of the input/output
    vector, the over-place DIF-FFT algorithm naturally computes
    $\brDFT_{N}(\cdot,\omega)$, while the over-place DIT-FFT algorithm on $\omega^{-1}$
    computes $N\cdot\brDFT_{N}^{-1}(\cdot, \omega)$.
\end{remark}

\begin{algorithm}[htbp]
\caption{In-place power of two accumulating multiplication.}\label{alg:fft2pow}
\begin{algorithmic}[1]
    \REQUIRE $\vec{a}$, $\vec{b}$ and $\vec{c}$ of respective lengths $n$, $n$ and $N = 2n$, containing the coefficients of $A$, $B$, $C\in\D[X]$ respectively; $\omega\in\D$ principal $N$-th root of unity, with $N = 2^p$.
\ENSURE $\vec{c}$ contains the coefficients of $C+A\cdot B$.
\STATE $\vec{c}\gets \brDFT_{2n}(\vec{c},\omega)$;
    \hfill\COMMENT{over-place}
\STATE\label{lin:DFTab}$\vec{a}\gets \brDFT_n(\vec{a},\omega^2)$; $\vec{b}\gets \brDFT_n(\vec{b},\omega^2)$
    \hfill\COMMENT{over-place}
\ForDoEnd{$i=0$ \To $n-1$}{$c_i\pe a_i\times b_i$}
\STATE $\vec{a}\gets \brDFT^{-1}_n(\vec{a},\omega^2)$; $\vec{b}\gets \brDFT^{-1}_n(\vec{b},\omega^2)$
    \hfill\COMMENT{undo \ref{lin:DFTab}}
\ForDoEnd[lin:mulw]{$i=0$ \To $n-1$}{$a_i\fe \omega^i$; $b_i\fe \omega^i$}
\STATE\label{lin:DFTab2}$\vec{a}\gets \brDFT_n(\vec{a},\omega^2)$; $\vec{b}\gets
\brDFT_n(\vec{b},\omega^2)$
    \hfill\COMMENT{over-place}
\ForDoEnd{$i=0$ \To $n-1$}{$c_{i+n}\pe a_i\times b_i$}
\STATE $\vec{a}\gets \brDFT^{-1}_n(\vec{a},\omega^2)$; $\vec{b}\gets \brDFT^{-1}_n(\vec{b},\omega^2)$
    \hfill\COMMENT{undo \ref{lin:DFTab2}}
\ForDoEnd{$i=0$ \To $n-1$}{$a_i\de \omega^i$; $b_i\de \omega^i$}
    \hfill\COMMENT{undo \ref{lin:mulw}}
\RETURN $\vec{c}\gets \brDFT^{-1}_{2n}(\vec{c},\omega)$
\end{algorithmic}
\end{algorithm}

\begin{theorem}\label{thm:fft2pow}
  Using an over-place $\brDFT$ algorithm with
  complexity bounded by $\bigO{n\log n}$, \cref{alg:fft2pow} is
  correct, in-place and has complexity bounded by $\bigO{n\log n}$.
\end{theorem}
\begin{proof}
    \cref{alg:fft2pow} follows the pattern of the standard FFT-based
    multiplication
    algorithm. Our goal is to compute $\brDFT_{2n}(A,\omega)$, $\brDFT_{2n}(B,\omega)$ and
    $\brDFT_{2n}(C,\omega)$,
    then obtain $\brDFT_{2n}(C+AB,\omega)$ and finally $C+AB$ using an inverse
    $\brDFT$. Computations on $C$ and then $C+AB$ are performed over-place using
    any standard over-place $\brDFT$ algorithm. The difficulty happens for $A$ and
    $B$ that are stored in length-$n$ arrays. We use the following property of
    the bit reversed order: for $k < n$, $[k]_{p} = 2[k]_{p-1}$, and for
    $k \ge n$, $[k]_{p} = 2[k-n]_{p-1}+1$. Therefore, the first $n$
    coefficients of $\brDFT_{2n}(A,\omega)$ are $(A(\omega^{2[0]_{p-1}}), \dotsc
    A(\omega^{2[n-1]_{p-1}})) = \brDFT_n(A,\omega^2)$. Similarly, the next $n$
    coefficients are $\brDFT_n(A(\omega X), \omega^2)$. Therefore, one can compute
    $\brDFT_n(A,\omega^2)$ and $\brDFT_n(B,\omega^2)$ in $\vec{a}$ and $\vec{b}$ respectively,
    and update the first $n$ entries of $\vec{c}$. Next we restore $\vec{a}$
    and $\vec{b}$ using $\brDFT^{-1}_n(\cdot,\omega^2)$. We compute $A(\omega X)$ and
    $B(\omega X)$ and again $\brDFT_n(A(\omega X),\omega^2)$ and $\brDFT_n(B(\omega X),\omega^2)$ to update the
    last $n$ entries of $\vec{c}$. Finally, we restore $\vec{a}$ and
    $\vec{b}$ and perform $\brDFT^{-1}$ on $\vec{c}$.
    The cost is dominated by the ten $\brDFT^{\pm1}$ computations.
\end{proof}

The standard (not in-place) algorithm uses two $\brDFT$ and one $\brDFT^{-1}$ in size $2n$.
Since our in-place variant uses 2 size-$2n$ and 8 size-$n$ $\brDFT^{\pm1}$, the dominant
term is twice as large.

The case where the sizes are not powers of two is loosely similar, using as a routine
a truncated Fourier transform (TFT) rather than a
DFT~\cite{2004:vanderHoeven:ISSAC:TFT}. Let $\omega$ still be an $N$-th root of
unity for some $N = 2^p$, and $n < N$.
The length-$n$ (bit-reversed) TFT of a polynomial $F\in\D[X]$ at $\omega$
is $\brTFT_n(F,\omega) = (F(\omega^{[0]_p}), \dotsc, F(\omega^{[n-1]_p}))$, that is
the $n$ first coefficients of $\brDFT_{N}(F,\omega)$.
As for the DFT, the (bit-reversed) TFT and its inverse
can be computed over-place~\cite{Harvey:2010:issactft, Roche:2011:waterloo,
2013:Arnold:ISSAC:TFT, Coxon:2022:JSC:inplaceTFT}.

Given inputs $A$ and $B\in\D[X]$ of respective lengths $m$ and $n$ and an output
$C\in\D[X]$ of length $m+n-1\leq N$, we aim to replace $C$ by $C+AB$. As in the power-of-two
case, we first replace $C$ by %
$\brTFT_{m+n-1}(C,\omega)$ in  $\vec{c}$.
Then we progressively update $\vec{c}$ using small $\brTFT$'s on the inputs, using the following lemma.

\begin{lemma}[\cite{Harvey:2010:issactft,Roche:2011:waterloo}]\label{lem:roche}
    Let $F\in\D[X]$, $\ell$, $s\in\Z{>0}$ %
    where $2^\ell$ divides
    $s$, and $\omega$ be a $2^p$-th principal root of unity. If $F_{s,\ell}(X) =
    F(\omega^{[s]_p}X)\bmod X^{2^\ell-1}$,
    \(\brDFT_{2^\ell}(F_{s,\ell},\omega^{2^{p-\ell}}) = (F(\omega^{[s]_p}),\dotsc,F(\omega^{[s+2^\ell-1]_p}))\).
\end{lemma}
\begin{proof}
    Let $\omega_\ell = \omega^{2^{p-\ell}}$. This is a principal $2^\ell$-th root of unity
    since $\omega$ is a principal $2^p$-th root of unity. In particular, for any $i <
    2^\ell$, $F_{s,\ell}(\omega_\ell^{[i]_\ell}) = F(\omega^{[s]_p}\omega_\ell^{[i]_\ell})$.
    Now, $\omega_\ell^{[i]_\ell} = \omega^{[i]_p}$ since $2^{p-\ell}[i]_\ell = [i]_p$.
    Furthermore, $[s]_p+[i]_p = [s+i]_p$ since $i < 2^\ell$ and $2^\ell$ divides
    $s$.
    Finally, $F_{s,\ell}(\omega_\ell^{[i]_\ell}) = F(\omega^{[s+i]_p})$.
\end{proof}

\begin{corollary}\label{cor:partTFT}
    Let $F\in\D[X]$ stored in an array $\vec{f}$ of length $n$,
    $\ell$, $k\in\Z_{>0}$ and $\omega$ be a $2^p$-th principal root of unity, with $2^\ell\le
    n$ and $(k+1)2^\ell \le 2^p$. There exists an algorithm,
    $\partTFT_{k,\ell}(\vec{f},\omega)$, that replaces the first $2^\ell$ entries of
    $\vec{f}$ by $F(\omega^{[k\cdot 2^\ell]_p})$, \dots,
    $F(\omega^{[(k+1)\cdot2^\ell-1]_p})$, and an inverse algorithm
    $\partTFT^{-1}_{k,\ell}$ that restores $\vec{f}$ to its initial state. Both
    algorithms are in-place have complexity bounded by
    $\bigO{n+\ell\cdot2^\ell}$.
\end{corollary}
\begin{proof}
    Algorithm $\partTFT_{k,\ell}(\vec{f}, \omega)$ is the following:
    \begin{algorithmic}[1]
    \ForDoEnd{$i=0$ \To $n-1$}{$f_i\fe \omega^{i[k\cdot2^\ell]_p}$}
    \ForDoEnd{$i=2^\ell$ \To $n-1$}{$f_{i-2^\ell}\pe f_i$}
    \STATE $\vec{f}_{0..2^\ell-1} \gets \brDFT_{2^\ell}(\vec{f}_{0..2^\ell-1},\omega^{2^{p-\ell}})$
    \end{algorithmic}
    Its correctness is ensured by~\cref{lem:roche}. Its inverse algorithm
    $\partTFT^{-1}_{k,\ell}(\vec{f},\omega)$ does the converse:
    \begin{algorithmic}[1]
    \STATE $\vec{f}_{0..2^\ell-1} \gets\brDFT^{-1}_{2^\ell}(\vec{f}_{0..2^\ell-1},\omega^{2^{p-\ell}})$
    \ForDoEnd{$i=2^\ell$ \To $n-1$}{$f_{i-2^\ell}\me f_i$}
    \ForDoEnd{$i=0$ \To $n-1$}{$f_i\de \omega^{i[k\cdot2^\ell]_p}$}
    \end{algorithmic}
    In both algorithms, the call to $\brDFT^{\pm1}$ has cost
    $\bigO{\ell\cdot{2^\ell}}$, and the two other steps have cost~$\bigO{n}$.
\end{proof}

To implement the previously sketched strategy, we assume that $m\le n$ for
simplicity. We let $\ell$, $t$ be such that $2^\ell\le m<2^{\ell+1}$ and
$2^{\ell+t}\le n<2^{\ell+t+1}$. Using $\partTFT^{\pm 1}$, we are able to compute
$(A(\omega^{[k\cdot 2^\ell]_p}), \dotsc, A(\omega^{[(k+1)\cdot 2^{\ell}-1]_p}))$ for any
$k$ and restore $A$ afterwards, and similarly for
\[(B(\omega^{[k\cdot 2^{\ell+t}]_p}), \dotsc, B(\omega^{[(k+1)\cdot 2^{\ell+t}-1]_p})).\]

\begin{algorithm}[htbp]
\caption{In-place fast accumulating polynomial multiplication.}\label{alg:fftaccu}
\begin{algorithmic}[1]
\REQUIRE $\vec{a}$, $\vec{b}$ and $\vec{c}$ of length $m$, $n$ and $m+n-1$,
$m\le n$, containing the coefficients of $A$, $B$, $C\in\D[X]$ respectively;
$\omega\in\D$ principal $2^p$-th root of unity with $2^{p-1} < m+n-1 < 2^p$
\ENSURE $\vec{c}$ contains the coefficients of $C+A\cdot B$.
\STATE $\vec{c}\gets \brTFT_{m+n-1}(\vec{c},\omega)$;
    \hfill\COMMENT{over-place}
\STATE $r\gets m+n-1$
\WHILE{$r \ge 0$}
\STATE\label{lin:boundslt}$\ell\gets\lfloor\log_2\min\{r, m\}\rfloor$; $t\gets\lfloor\log_2\min\{r, n\}\rfloor-\ell$;
 \STATE\label{lin:boundsk}$k\gets m+n-1-r$
 \hfill\COMMENT{over-place: $B(\omega^{[k\cdot2^{\ell+t}]_p]}), \dotsc, B(\omega^{[(k+1)\cdot 2^{\ell+t}-1]_p})$}
    \STATE\label{lin:TFTb}$\vec{b}\gets\partTFT_{k,\ell+t}(\vec{b},\omega)$
	\FOR{$s=0$ \To $2^t-1$}\hfill\COMMENT{over-place: $A(\omega^{[(k\cdot2^t+s)2^\ell]_p]}), \dotsc, A(\omega^{[(k\cdot 2^t+s+1)2^\ell-1]_p})$}
	\STATE\label{lin:TFTa}$\vec{a}\gets\partTFT_{s+k\cdot2^t,\ell}(\vec{a},\omega)$
	\ForDoEnd{$i=0$ \To $2^\ell-1$}{$c_{i+(k\cdot 2^t+s)2^\ell} \pe  a_i b_{i+s\cdot 2^\ell}$}
	\STATE $\vec{a}\gets\partTFT^{-1}_{s+k\cdot2^t,\ell}(\vec{a},\omega)$
	\hfill\COMMENT{undo~\ref{lin:TFTa} over-place}
	\ENDFOR
    \STATE $\vec{b}\gets\partTFT^{-1}_{k,\ell+t}(\vec{b},\omega)$
    \hfill\COMMENT{undo~\ref{lin:TFTb} over-place}
    \STATE $r\me 2^{\ell+t}$
\ENDWHILE
\RETURN $\vec{c}\gets \brTFT^{-1}_{m+n-1}(\vec{c},\omega)$
\end{algorithmic}
\end{algorithm}

\begin{theorem}
    \Cref{alg:fftaccu} is correct and in-place.  If the algorithm
    $\brDFT$ used inside $\partTFT$ has complexity $\bigO{n\log n}$, then the running
    time of~\cref{alg:fftaccu} is $\bigO{n\log n}$.
\end{theorem}

\begin{proof}
    The fact that the algorithm is in-place comes
    from~\cref{cor:partTFT}.
    The only slight difficulty is to produce, fast
    and in-place, the relevant roots of unity. This is actually dealt with in
    the original over-place TFT algorithm~\cite{Harvey:2010:issactft} and can be
    done the same way here.

    To assess its correctness, first note that the values
    of~\cref{lin:boundslt,lin:boundsk}
    are computed so that $2^\ell\le r,m$ and
    $2^{\ell+t}\le r,n$. One iteration of the while loop updates the entries
    $c_k$ to $c_{k+2^{\ell+t}-1}$ where $k = m+n-1-r$. To this end,
    we first compute $B(\omega^{[k\cdot2^{\ell+t}]_p]})$ to $B(\omega^{[(k+1)\cdot
    2^{\ell+t}-1]_p})$ in $\vec{b}$ using $\partTFT$. Then, since $\vec{a}$ may
    be too small to store $2^{\ell+t}$ values, we compute the corresponding
    evaluations of $A$ by groups of $2^\ell$, using a smaller $\partTFT$. After
    each computation in $\vec{a}$, we update the corresponding entries in
    $\vec{c}$ and restore $\vec{a}$. Finally, at the end of the iteration,
    entries $k$ to $k+2^{\ell+t}-1$ of $\vec{c}$ have been updated and $\vec{b}$
    can be restored. This proves the correctness of the algorithm.

    To bound its complexity, %
    we first bound the number of iterations of the
    while loop. We identify two phases, first iterations where $r \ge n$ and
    then iterations with $r < n$. The first phase has at most $3$ iterations since
    $2^{\ell+t}>\frac{n}{2}$ entries of $\vec{c}$ are updated per iteration.
    The second phase starts with $r < n$ and each iteration updates $2^{\ell+t}>\frac{r}{2}$ entries.
    That is, $r$ is halved and this
    second phase has at most $\log_2{n}$ iterations.
    The cost of an iteration is dominated by the calls to $\partTFT^{\pm 1}$. The cost of
    a call to $\partTFT^{\pm 1}_{k,\ell}$ with a size-$m$ input is the sum of a
    linear term $\bigO{m}$ and a non-linear term $\bigO{\ell\cdot 2^\ell}$.
    At each iteration, there are two calls to $\partTFT^{\pm 1}$ on $\vec{b}$ and
    $2^{t+1}$ calls to $\partTFT^{\pm 1}$ on $\vec{a}$. The linear terms sum to
    $\bigO{n+m\cdot 2^t} = \bigO{n}$ since $m\cdot
    2^t < 2^{\ell+1+t} \le 2n$. Over the $\log_2{n}$ iterations, the global
    cost due to these linear terms is $\bigO{n\log n}$.
    The cost due to the non-linear terms in one iteration is
    $\bigO{(\ell+t)\cdot{2^{\ell+t}}}$. In the first iterations,
    $2^{\ell+t} \le n$ and these costs sum to $\bigO{n\log n}$.
    In the next iterations, $2^{\ell+t} \le r < n$. Since $r$ is
    halved at each iteration, the non-linear costs in these
    iterations sum to
    $\bigOdisplay{\sum_i\frac{n}{2^i}\log\frac{n}{2^i}}=\bigO{n\log{n}}$.
\end{proof}

Then, \cref{alg:fftaccu} is compared with previous FFT-based
algorithms for polynomial multiplications designed to reduce their
memory footprint in~\cref{tab:fft} (see also
\cite[Table~2.2]{Giorgi:2019:hdr}).
Note that no call stack is needed for computing the FFT, therefore
these algorithms only require $\bigO{1}$ pointer registers.

\begin{table}[ht]\centering
\caption{Reduced-memory algorithms for FFT polynomial
  multiplication}\label{tab:fft}
\begin{tabular}{lccc}
\toprule
Algorithm & Algebraic Reg. & Inputs & Accumulation\\
\midrule
\cite{1965:CooleyTukey:MathComp:FFT} & $2n$ & {\color{teal} read-only} & {\xno}\\
\cite{Roche:2009:spacetime} & $\bigO{2^{\lceil\log_2 n\rceil}-n}$ & {\color{teal} read-only} & {\xno}\\
\cite{Harvey:2010:issactft} & \textcolor{teal}{$\bigO{1}$} & {\color{teal} read-only} & {\xno}\\
\cref{alg:fftaccu} & \textcolor{teal}{$\bigO{1}$} & mutable & {\xyes}\\
\bottomrule
\end{tabular}
\end{table}
\section{Conclusion}
We here provide a generic technique mapping any bilinear formula
(and more generally any linear accumulation)
into an in-place algorithm.
This allows us for instance to provide the first accumulating in-place
Strassen-like matrix multiplication algorithm.
This also allows use to provide fast in-place accumulating polynomial
multiplications algorithms.

Many further accumulating algorithm can then be reduced to these fundamental
building blocks, see for instance Toeplitz, circulant, convolutions or
remaindering operations in~\cite{jgd:2023:inplrem}.

\renewcommand\thefootnote{}
\footnotetext{\emph{This material is based on work
supported in part by the Agence Nationale pour la Recherche under Grants
\href{https://anr.fr/Project-ANR-21-CE39-0006}{ANR-21-CE39-0006},
\href{https://anr.fr/ProjetIA-15-IDEX-0002}{ANR-15-IDEX-0002} and
\href{https://anr.fr/ProjetIA-22-PECY-0010}{ANR-22-PECY-0010}.}}

%
%
%

%
\bibliographystyle{plainurl}
\bibliography{shortbib}
\end{document}